\documentclass[11pt]{article}

\usepackage{authblk}

\title{Communication-Optimal Parallel Standard and Karatsuba Integer Multiplication in the Distributed Memory Model}

\author{Lorenzo De Stefani\thanks{lorenzo\_destefani@brown.edu}}
\affil{Department of Computer Science, Brown University}

 \usepackage[margin=1.25in]{geometry}

\AtBeginDocument{%
  \providecommand\BibTeX{{%
    \normalfont B\kern-0.5em{\scshape i\kern-0.25em b}\kern-0.8em\TeX}}}

\usepackage{amsmath,amssymb,amsfonts}

\usepackage{enumitem}
\usepackage{graphicx}
\usepackage{textcomp}
\usepackage{xcolor}
\usepackage{pgf}
\usepackage{xcolor} 
\usepackage{tikz}
\usetikzlibrary{fit,calc}

\colorlet{pink}{red!40}
\colorlet{blue}{cyan!60}
\usetikzlibrary{arrows, automata, petri}
\usepackage{longtable}
\usepackage{todonotes}
\usepackage[ruled, vlined, linesnumbered]{algorithm2e}

\usepackage{amsthm}
\usepackage{mathrsfs}

\newcommand{\BO}[1]{\mathcal{O}\left(#1\right)}
\newcommand{\BOme}[1]{\Omega\left(#1\right)}

\newcommand{\io}{I/O}

\newcommand{\alg}{\mathcal{A}}
\newcommand{\nproc}{\mathcal{P}}
\newcommand{\npr}[1]{|\mathbf{#1}|}

\newcommand{\mymax}[1]{\max\left\{#1\right\}}
\newcommand{\mymaxo}[2]{\max_{#1}\left\{#2\right\}}

\newcommand{\msgsize}{B_m}

\makeatletter
\newtheorem*{rep@theorem}{\rep@title}
\newcommand{\newreptheorem}[2]{%
\newenvironment{rep#1}[1]{%
 \def\rep@title{#2 \ref{##1}}%
 \begin{rep@theorem}}%
 {\end{rep@theorem}}}
\makeatother
\newreptheorem{theorem}{Theorem}

\newtheorem{theorem}{Theorem}
\newtheorem{lemma}[theorem]{Lemma}
\newtheorem{fact}[theorem]{Fact}

\newcommand{\coim}{\textsc{COPSIM}}
\newcommand{\cok}{\textsc{COPK}}

\allowdisplaybreaks

\begin{document}
\begin{titlepage}
\pagenumbering{gobble}
\maketitle



\begin{abstract}
  We present \coim{} a parallel implementation of standard integer multiplication for the distributed memory setting, and \cok{} a parallel implementation of Karatsuba's fast integer multiplication algorithm for a distributed memory setting. 
When using $\mathcal{P}$ processors, each equipped with a \emph{local memory}, to compute the product of tho $n$-digits integer numbers, under mild conditions, our algorithms achieve \emph{optimal speedup} of the computational time. That is, $\BO{n^2/\mathcal{P}}$ for \coim{}, and $\BO{n^{\log_2 3}/\mathcal{P}}$ for \cok{}. The total amount of memory required across the processors is $\BO{n}$, that is, within a constant factor of the minimum space required to store the input values.
We rigorously analyze the \emph{Input/Output} (\io{}) cost of the proposed algorithms. We show that their \emph{bandwidth cost} (i.e., the number of memory words sent or received by at least one processors) matches asymptotically corresponding known \io{} lower bounds, and their \emph{latency} (i.e., the number of messages sent or received in the algorithm's \emph{critical execution path}) is asymptotically within a multiplicative factor $\BO{\log^2_2 \mathcal{P}}$ of the corresponding known \io{} lower bounds. Hence, our algorithms are asymptotically optimal with respect to the bandwidth cost and almost asymptotically optimal with respect to the latency cost. 
\end{abstract}

\end{titlepage}
\pagenumbering{arabic}
\section{Introduction}\label{sec:intro}
Integer multiplication is a widely used and widely studied basic primitive with many important applications, among which primes factorization is of particular notice due to its impact on the field of cryptography~\cite{jahani2014efficient, mera2020time}.
The importance of integer multiplication can be fully appreciated by noting many computers implement it in hardware. Still, it is also complex enough that in many other very successful cases, it is entirely computed by software.

The standard algorithm
(also known as the long multiplication or the schoolbook algorithm)
takes $\Theta\left(n^2\right)$ digit operations to multiply two
$n$-digit numbers.  In 1960,
Karatsuba~\cite{karatsuba1962multiplication} showed how to improve the
bound to $\Theta\left(n^{\omega}\right)$, where $\omega = \log_2
3 \approx 1.585$.  This result has motivated a number of efforts which
have led to increasingly faster algorithms. Among these, of particular note are the \emph{Toom-Cook}
algorithmic scheme  originally introduced by Andrei
Toom~\cite{toom1963complexity} for circuits and later adapted by
Stephen Cook~\cite{cook} to software programs, the asymptotically
faster $\Theta\left(n \log n \log \log n \right)$ 
 Sch\"onhage-Strassen algorithm~\cite{schonhage1971schnelle}, and F\"urer's algorithm~\cite{furer2009faster} with complexity $\Theta\left(n \log n 2^{\BO{\log^*n}} \right)$, where $\log^*n$ is the iterated logarithm, and, most recently, the algorithm by Harvey and van der Hoven~\cite{harvey2019faster} with complexity $\Theta\left(n \log n\right)$.
 However, due to the, sometimes extremely high, constant multiplicative factors ``\emph{hidden}'' by the asymptotic notation, standard-long integer multiplication and Karatsuba's algorithm actually outperform the other, asymptotically faster, algorithms for a wide range of input sizes up to $2^{2^{14}}$~\cite{garcia2005can}. Hence, both standard and Karasuba's algorithms are of great practical interest. The problem of improving the performance of integer multiplication algorithms is actively researched, as evidenced by the significant number of publications in this field. 
 
 While promising, designing parallel algorithms based on the known fast multiplication algorithms appears challenging due to the apparent ``\emph{sequential nature}'' of the integer multiplication algorithms discussed so far, and the necessity to carefully manage communications among the processors participating in the computation \emph{around} such sequential components.

When designing efficient parallel algorithms, it is important not only to balance the computational effort among processors but also to minimize the time spent by the processors communicating to each other to transfer data and coordinate operations. The  \emph{communication cost} (or \emph{I/O} cost) is, in many cases, much higher than that due to computation, and, therefore, is the real bottleneck of algorithmic performance. This technological
trend~\cite{patterson2005getting} appears destined to continue, as physical
limitations on minimum device size and maximum message speed lead
to inherent costs when moving data, whether across the levels of a
hierarchical memory system or between processing elements of a
parallel system~\cite{bilardi1995horizons}. Due to these challenges, most parallel algorithms for integer multiplication were proposed for the \emph{shared memory model} where all processors have access to a shared memory space (among others, ~\cite{mansouri2014parallelization,Kronenburg2016ToomCookMS,giorgi2013parallel,jebelean1997using,Tembhurne2014PerformanceEO,edamatsu2018acceleration}). While this model simplifies many of the mentioned challenges related to communication, it is rather unrealistic for modern architectures. 
In this work, we consider a more realistic parallel distributed-memory model, where each of the $\mathcal{P}$ processor is equipped with a local (non-shared) memory space which can hold up to $M$ memory words, and data communication among processors occurs only by message exchange. 

Other approaches have been presented for specific hardware devices (e.g., FPGA)~\cite{pan2019hardware,portugal2006reversible,bahigfast,tembhurne2017parallel} or for models with limitations in the number of available processors. While some parallel versions of the standard and Karatsuba algorithms were presented in the literature for the distributed memory model, these contributions focus on specific settings with respect to the number of available processors~\cite{cesari1996performance, char1994some}, or assume unbounded local memory space~\cite{knuth1981seminum}. Further, in all mentioned contributions, the impact of the communication over execution time is evaluated through experimental evaluation of specific implementations rather than a formal theoretical analysis, or a rigorous comparison with theoretical lower bounds. 

In recent contributions, De Stefani~\cite{stefani2019io}, and Bilardi and De Stefani~\cite{BilardiS19} presented the first analytical lower bound on the communication cost of, respectively, standard integer multiplication algorithms, and Toom-Cook fast integer multiplication, of which Karatsuba can be seen as a special case. Their results for the distributed-memory parallel model yield lower bounds for both the \emph{bandwidth cost} (i.e., the number of memory words transmitted by at least one processor) and the \emph{latency} (i.e., the number of messages exchanged by at least one processor) of parallel integer multiplication algorithms. These works left open the important question of whether it is actually possible to construct algorithms matching these bounds.

In this work, we present \coim{} a parallel algorithm based on the recursive long multiplication algorithm, and \cok{} a parallel fast integer multiplication algorithm based on Karatsuba's algorithm. Both our algorithms are designed for the distributed memory model. Under very mild conditions (i.e., $n\geq \mathcal{P}$ and $M\geq \log_2 \mathcal{P}$), our algorithms achieve \emph{optimal speedup} of the computation time with respect to their sequential counterpart, asymptotically optimal communication bandwidth cost, and latency within a $\BO{\log^2 \mathcal{P}}$ multiplicative factor of the corresponding lower bounds~\cite{BilardiS19,stefani2019io}. Finally, both our algorithms require only $\BO{n}$ memory space to be available when combining the size of the local memories of the processors. That is, the total required memory space is within a constant multiplicative factor of the memory space required to store the input. Hence, \coim{} and \cok{} have \emph{asymptotically optimal memory requirement}.
To the best of our knowledge, ours are the first parallel algorithms for integer multiplication to achieve computational and communication optimality in the distributed memory setting. 

\paragraph{Related work}
As discussed in the introduction, various parallel implementation of standard-long integer multiplication algorithms have been presented in the literature for the shared memory model (among others, ~\cite{mansouri2014parallelization,Kronenburg2016ToomCookMS,giorgi2013parallel,jebelean1997using,Tembhurne2014PerformanceEO,edamatsu2018acceleration}), and for for specific hardware (among others, ~\cite{pan2019hardware,portugal2006reversible,bahigfast,tembhurne2017parallel}). The analysis of the communication component of these algorithms' execution time is mostly given as experimental evaluation of specific implementations of the proposed algorithms rather than a formal analysis of their scalability for a range of values if input size, number of available processors, and available memory. In this work, we present a rigorous analysis of the computation time, memory requirement, and communication cost for both our proposed algorithms.

Similarly, parallel versions of the Karatsuba's algorithm are mostly presented for the shared memory setting~\cite{kuechlin1991integer}, or focus on the experimental analysis of specific implementations without formally analyzing the scalability and the communication cost of the proposed algorithms~\cite{char1994some}. 
In~\cite{cesari1996performance}, Cesari and Maeder introduce three parallel Karatusba-based algorithms for the distributed memory setting: The first two algorithms have time complexity $O(n)$, where $n$ denotes the number of digits of the input integers, when using $n^{\log_2 3}$ processors. The last one exhibits $O(n \log n)$ time complexity while using $n$ processors. Their approach follows an \emph{master-slave} approach where single processors are assigned recursively-generated subproblems to be solved in parallel, and they may themselves use other, still unused processors to do so. Thus, the scalability of these algorithms is limited by the fact that long integer additions and subtractions need to be computed by single processors. Further, their approach does not account for limitations due to the size of the local memory available to the processors being used, as several processors need to store integer values of size $O(n)$ entirely. 
In contrast, our~\cok{} algorithm achieves computational time $\BO{n^{\log_2 3}/\mathcal{P}}$ for any number $\mathcal{P}\leq n$ of processors available processors. Both the computational and the communication cost of \cok{} scales proportionally with $1/\mathcal{P}$, thus exhibiting perfect strong scaling. Further, the \emph{comulative memory space} across the processors required by the algorithm is within a constant factor of that necessary to represent the input factor integers and their product.

The analysis of the communication requirement of algorithms has been studied extensively in the literature both in the sequential and the parallel setting. There have been also numerous efforts to obtain communication efficient parallel algorithms for many important problems among whom the computation of the FFT~\cite{chu1998fft}, Cholesky decomposition~\cite{ballard2010communication, liu1986computational}, Matrix Factorization~\cite{o1986assignment}, and Matrix Multiplication~\cite{ballard2012communicationalg,ballard2011minimizing,berntsen1989communication,choi1994pumma,cannon1969cellular,agarwal1995three}.
In particular, in~\cite{ballard2012communicationalg} Ballard et. al presented \emph{CAPS} a parallel version of Strassens's algorithm for fast matrix multiplication~\cite{strassen1969gaussian}. Their algorithm achieves optimal speedup and it minimizes the bandwidth cost among all parallel Strassen-based algorithms. This work draws inspiration from the technique used in their work to obtain communication-optimal algorithms for integer multiplication.  Doing so requires several major, and challenging, modifications due to differences between matrix and integer multiplication, and, in particular, the apparently sequential nature of components of the latter, which our algorithms overcome by speculatively precalculating some intermediate results of the algorithm.

As mentioned in the introduction, De Stefani~\cite{stefani2019io} and Bilardi and De Stefani~\cite{BilardiS19}, presented, respectively, lower bounds on the communication complexity of parallel implementations of standard-long integer multiplication algorithms~\cite{stefani2019io}, and of Toom-Cook algorithms~\cite{BilardiS19} in a memory-distributed model. We present these bounds in detail in Section~\ref{sec:bounds}, and they will serve as a term of comparison when evaluating the performance of our proposed algorithms. 

\paragraph{Our contributions}
We present two parallel integer multiplication algorithms for the distributed memory setting called \coim{} (Communication Optimal Parallel Standard Integer Multiplication) and ~\cok{} (Communication Optimal Parallel Karatsuba):
\begin{itemize}
    \item \coim{} computes the product of two given $n$ digits input integers using $\mathcal{P}$ processors (with $n\geq \mathcal{P}$) each equipped with a local memory of size $M\geq 24\sqrt{\mathcal{P}}$, in $\BO{n^2/\mathcal{P}}$ parallel computational steps. \coim{} exhibits  $\BO{n^2/(MP)}$ bandwidth cost and $\BO{n^2/(M^2\mathcal{P})}$ latency. Thus, by the known lower communication lower bounds in~\cite{stefani2019io}:
    \begin{theorem}[Communication optimality of \coim{}]\label{thm:informalcoim}
\coim{} achieves optimal computation time speedup and optimal bandwidth cost among all parallel standard integer multiplication algorithms. It also minimizes the latency cost up to a $\BO{\log^2\mathcal{P}}$ multiplicative factor. 
\end{theorem}
    \item \cok{} computes the product of two given $n$ digits input integers using $P$ processors (with $n\geq \mathcal{P}$) each equipped with a local memory of size $M \geq 10\mathcal{P}^{(\log_3 3)/2}$, in $\BO{n^{\log_2 3}/\mathcal{P}}$ parallel computational steps. \cok{} exhibits  $\BO{\left(\frac{n}{M}\right)^{\log_2 3}\frac{M}{\mathcal{P}}}$ bandwidth cost and $\BO{\left(\frac{n}{M}\right)^{\log_2 3}\frac{1}{\mathcal{P}}}$ latency. Thus, by the known lower communication lower bounds in~\cite{BilardiS19,stefani2019io}:
    \begin{theorem}[Communication optimality of \cok{}]\label{thm:informalcok}
\cok{} achieves optimal computation time speedup and optimal bandwidth cost among all parallel Karatsuba-based integer multiplication algorithms. It also minimizes the latency cost up to a $\BO{\log^2\mathcal{P}}$ multiplicative factor. 
\end{theorem}
\end{itemize}

Both our algorithms require $\BO{n}$ total memory space to be available across all processors. That is, each of the $\mathcal{P}$ processors requires a local memory of size only $\BO{n/\mathcal{P}}$. That is, the required memory space is within a constant factor of the minimum memory space necessary to store the input (and output) values. Both \coim{} and \cok{} are \emph{strongly scaling} as both the computation time and the bandwidth cost scale linearly with respect to $\mathcal{P}^{-1}$, provided that the size of the local memory of each processor scales accordingly (i.e., it is $\BO{n/\mathcal{P}}$).

A rigorous analysis of the performance of \coim{} (resp., \cok{}) is given in Theorem~\ref{thm:costcoimUM} and Theorem~\ref{thm:coim} in Section~\ref{sec:algdef} (resp., Theorem~\ref{thm:costcokUM} and Theorem~\ref{thm:cok} in Section~\ref{sec:karatsuba}). Proof of Theorem~\ref{thm:informalcoim} is presented in Section~\ref{sec:stalwbcomp}, and proof Theorem~\ref{thm:informalcok} is presented in Section~\ref{sec:karlwbcomp}. Further, these results imply that the lower bounds on the bandwidth cost of parallel standard-long inter multiplication algorithms~\cite{stefani2019io} and for parallel Karatsuba's algorithms~\cite{BilardiS19,stefani2019io} are indeed \emph{asymptotically tight}.

Our methods use a recursive divide-and-conquer approach and speculatively precalculate multiple possible values that may be used in the continuation of the algorithm in order to overcome the challenges related to the apparently sequential nature of integer multiplication algorithms. While this may seem wasteful in computation time and usage of available computational resources, this allows us to exploit the available parallelism while incurring low computational overhead. 

Our algorithms are designed with the intent of making the \emph{best possible use} of the memory space available to each processor. 
This is achieved by analyzing the recursion tree corresponding to the algorithm's execution, and by scheduling its traversal using an opportune combination of \emph{depth-first} and \emph{breadth-first} steps as discussed in Section~\ref{sec:overview}.

\paragraph{Paper organization}
In Section~\ref{sec:preliminaries}, we present an overview of the notation and the computational model considered in this work. In Section~\ref{sec:bounds}, we present an overview of the known lower bounds on communication cost of integer multiplications, which will serve as a term of comparison in evaluating the performance of our proposed algorithms.   In Section~\ref{sec:overview}, we present the main common strategy used by both our proposed algorithm. In Section~\ref{sec:components}, we present subroutines for adding, comparing, and subtracting integers using multiple processors in the distributed memory setting. These subroutines are used extensively in our algorithms. We present and fully analyze \coim{} (Communication Optimal Integer Multiplication) in Section~\ref{sec:algdef}, and \cok{} (Communication Optimal Karatsuba) in Section~\ref{sec:karatsuba}.

\section{Preliminaries}\label{sec:preliminaries}
We discuss algorithms that compute the product of two integers: $C = A\times B$. We assume the input integers to be expressed as a sequence of $n$ base-$s$ digits in \emph{positional notation}. We further assume that each integer is represented as an unsigned integer with an additional bit to denote the sign. For a given integer $A$, we denote its expansion in base $s$ as:
\begin{equation*}
    A = \left(A[n-1],A[n-2],\ldots A[0]\right)_s,
\end{equation*}
where $n$ is the number of digits in the base-$s$ expansion of $A$, and its digits are indexed in order from the least significant digit~$A[0]$ to the most significant digit~$A[n-1]$. Further, for $i\in\{1,1,\ldots,n\}$, we refer to $A[i-1]$ (resp., $A[n-i]$) as the $i$-th \emph{least} (resp., \emph{most}) significant digit of $A$. 
With a slight abuse of notation, we use $A$ (resp., $B$) to denote both the value being multiplied and the set of input variables to the algorithm.  

We refer to the number of digits of the base-$s$ expansion of an integer $A$ as its  ``\emph{size}'', and we denote it as $|A|$.

We consider parallel algorithms for integer multiplication in a distributed-memory parallel model where $P$ processors, each equipped with local (non shared) memory that can hold up to $M$ memory words, are directly connected to each other by a network. Each processor in the model is identified by an \emph{unique code} given by an integer value from $\{0,1,\ldots,P-1\}$. 
Processors can exchange point-to-point messages, with every message containing up to $\msgsize{}$ memory words. In the following, we refer to the number of memory words which can be stored in the local memory (resp., that can be transmitted in a single message), as the \emph{size} of the memory (resp., of the messages). 
We assume that each processor is equipped with digit-wise product and algebraic sum elementary operations. Further, we assume that the processor is equipped with operations for producing the most and least significant digits of an integer in base $s$. Unless explicitly stated otherwise, when referring to the ``\emph{digits}'' of integers, we mean the digits of their expansion in the base chosen for their representation in memory.

\subsection{Data layout}
We assume both the input integers and intermediate results to be stored in memory expressed as their base-$s$ expansion, with $s\in \mathbb{N}^{+}$, and with $s$ being at most equal to the maximum value which can be maintained in a single memory word plus one. (That is, if a memory word can fit 32 bits, we have $2\leq s \leq 2^{32}-1$.) In particular, we assume each digit in the base-$s$ expansion of a value to be stored in a different memory word. 

Given a set of available processors, in this work, we will often consider them organized in \emph{ordered sequences}.
An ordered sequence of processors $\mathbf{P}=\left(P_{\npr{P}-1},\ldots,P_0\right)$, we denote as $\mathbf{P[i]}$ the $i$-th processor in the sequence (indexed from the end), for $0\leq i\leq |\mathbf{P}|-1$. That is, if $\mathbf{P}=\left(P_{z},P_{y},\ldots,P_b,P_a\right)$, then $\mathbf{P}[0]= P_a$, $\mathbf{P}[1]= P_b$, and $\mathbf{P}[\npr{P}-1]= P_z$. Such ordered sequences will be used extensively through the presentation to clarify the organization of the processors in the computation, the assignment of digits of the same integer value across the local memory of multiple processors, and the patterns of communication among the processors. 

Given an integer $A$, the digits of its base-$s$ representation may be stored in the local memories of different processors. Given an ordered sequence of processors $\mathbf{P}$, we say that an $n$-digit integer $A$ is ``\emph{partitioned among the processors in $\mathbf{P}$ in $n'$ digits}'' if , for $0\leq j \leq |\mathbf{P}|-1$, the $n'$ digits of $A$ form the $\left(jn'\right)$-least significant (if any) to the $\left((j+1)n'\right)$-least significant are stored in the local memory of processor $\mathbf{P}[j+1]$. If $n\geq n'\npr{P}$, the remaining digits of $A$ are stored in the local memory of $\mathbf{\npr{P}-1}$. Sometimes we use the shorter expression ``$A$ is partitioned  in $\mathbf{P}$'', which implies $n'= \lceil n/\npr{P}\rceil$. 
When the digits of an integer $A$ are distributed among multiple processors, we assume that their digits are stored in \emph{positional notation} in the local memories of each of these processors.
In the following, we use the notation $A_{P_j}$ to denote the integer value whose base-$s$ expression corresponds to the digits of $A$ stored in the local memory of $P_j$. 

\subsection{Algorithmic performance metrics}
We characterize the performance of the proposed algorithms according to the following metrics:
\begin{itemize}
    \item The \emph{Memory requirement} $M(n,\mathcal{P})$, which denotes the memory space used in the local memory available to each processor;
    \item The \emph{Computational cost} $T(n,\mathcal{P},M)$, which denotes the number of digit-wise computations executed by during the algorithm's execution;
    \item The  \emph{Bandwidth} cost $BW(n,\mathcal{P},M)$, which denotes the number of memory words exchanged during the algorithm's execution;
    \item The  \emph{Latency} cost $L(n,\mathcal{P},M)$, which denotes the number of point-to-point words exchanged during the algorithm's execution;
\end{itemize}
where $n$ denotes the number of digits of the integers being multiplied, $\mathcal{P}$ denotes the number of processors being utilized, and $M$ denotes the size (in terms of memory words) of the local memory available to each processor.
We count the number of digit-wise operations, memory words exchanged,  and messages exchanged along the \emph{critical execution path} of the algorithm as defined by Yang and Miller~\cite{yang1988critical}. That is, operations executed in parallel by distinct processors are counted only once. Similarly, messages (and, thus, memory words) exchanged in parallel between distinct pairs of processors are counted only once. We assume that in any execution step of the algorithm a processor may only either \emph{send} or \emph{receive} a message to/from another processor but not both. 

These metrics can be composed to characterize the \emph{execution time of the algorithm}. Assume that the processors are \emph{homogeneous}, that is,  time $\alpha$ is required to compute a single digit-wise operation for each processor, and for each pair of processors the communication latency is $\beta$ and $\gamma$ time is required to transmit a memory word.
Then the overall execution time of the algorithm can be bound as:
\begin{equation*}
    \alpha T(n,\mathcal{P},M) + \beta L(n,\mathcal{P},M) + \gamma BW(n,\mathcal{P},M).
\end{equation*}

While the values of the constants $\alpha$,$\beta$ and $\gamma$ depend on the specific hardware being used, our analysis holds for any device and any network being used.

\subsection{Communication lower bounds for integer multiplication algorithms}\label{sec:bounds}
\paragraph{Lower bounds for parallel standard integer multiplication algorithms}
In~\cite{stefani2019io}, De Stefani introduced the following lower bounds on the communication costs of any parallel standard-long integer multiplication algorithm for a model analogous to that considered in this work and discussed at the beginning of Section~\ref{sec:preliminaries}.

\begin{theorem}[{\cite{stefani2019io}[Corollary 8]}]\label{thm:lwbstamemdip}
Let $\mathcal{A}$ be any standard integer multiplication algorithm which computes $\BOme{n^2}$ digit operations to multiply two integers $A,B$ represented as $n$-digit base-$s$ numbers using $\nproc{}$ processors each equipped with a local memory of size $M < n$. Then:
\begin{align*}
    BW(n,P,M) &=\BOme{\frac{n^2}{M\nproc{}}},\\
    L(n,P,M) &=\BOme{\frac{n^2}{M^2\nproc{}}}.
\end{align*}
\end{theorem}

In the same work, the author also provides \emph{memory independent} \io{} lower bounds which hold under the assumption that the input integers are originally distributed in a \emph{balanced} way among the processors, but require no assumption of the size of the local memories available to each processor:

\begin{theorem}[{\cite{stefani2019io}[Corollary 12]}]\label{thm:lwbstamemind}
Let $\mathcal{A}$ be any standard integer multiplication algorithm to multiply two integers $A,B$ represented as $n$-digit base-$s$ numbers using $\nproc{}$ processors each equipped with an unbounded local memory. Assume further that at the beginning of $\mathcal{A}$ no processor has more than $\alpha n/\nproc{}$ (where $\alpha$ is a constant with respect to $n/\nproc{}$) digits of each of the input integers stored in its local memory. Then:
\begin{align*}
	BW(n,\nproc{})  &\geq \BOme{\frac{n}{\msgsize{}\sqrt{\nproc{}}}},\\
	L(n,\nproc{})  &\geq \BOme{1}.
\end{align*} 
\end{theorem}

The memory-independent bound is dominant for $M\geq \BOme{\frac{n}{\msgsize{}\sqrt{\nproc{}}}}$, while the memory-dependent bound is dominant for $M\geq \BO{\frac{n}{\msgsize{}\sqrt{\nproc{}}}}$.

\paragraph{Lower bounds for parallel Karatsuba-based algorithms}
In~\cite{BilardiS19}, Bilardi and De Stefani introduced lower bounds on the communication costs of Toom-Cook integer multiplication algorithms for a model analogous to that considered in this work and discussed at the beginning of Section~\ref{sec:preliminaries}. As argued in their work, Karatsuba's algorithm can be seen as a special case of Toom-$2$. Hence, here we present a bound on the communication cost of Karatsuba's algorithm based on their more general result for Toom-Cook:

\begin{theorem}[{\cite{BilardiS19}[Theorem 4.2]}]\label{thm:karabound}
Let $\alg$ be an algorithm which uses the Karatsuba recursive strategy to compute the product of two integers $A,B$ whose base-$s$ expansions have $n$ digits using  $\nproc$ processors each equipped with a local memory of size $M$. Then:
\begin{align*}
    BW(n,P,M) &=\BOme{\left(\frac{n}{M}\right)^{\log_2 3}\frac{M}{\nproc{}}},\\
    L(n,P,M) &=\BOme{\left(\frac{n}{M}\right)^{\log_2 3}\frac{1}{\nproc{}}}.
\end{align*}
\end{theorem}

In~\cite{stefani2019io}, De Stefani provides \emph{memory independent} \io{} lower bounds for parallel Toom-Cook algorithm which hold under the assumption that the input integers are originally distributed in a \emph{balanced} way among the processors, but require no assumption of the size of the local memories available to each processor. Once again, here we present a bound on the communication cost of Karatsuba's algorithm based on the more general result for Toom-Cook in\cite{stefani2019io}:

\begin{theorem}[{\cite{stefani2019io}[Corollary 12]}]\label{thm:karaboundmi}
Let $\mathcal{A}$ be any Karatsuba-based integer multiplication algorithm to multiply two integers $A,B$ represented as $n$-digit base-$s$ numbers using $\nproc{}$ processors each equipped with an unbounded local memory. Assume further that at the beginning of $\mathcal{A}$ no processor has more than $\alpha n/\nproc{}$ (where $\alpha$ is a constant with respect to $n/\nproc{}$) digits of each of the input integers stored in its local memory. Then:
\begin{align*}
	BW(n,\nproc{})  &\geq \BOme{\frac{n}{P^{1/\log_2 3}}},\\
	L(n,\nproc{})  &\geq \BOme{1}.
\end{align*} 
\end{theorem}

The memory-independent bound is dominant for $M\geq \BOme{\frac{n}{\nproc{}^{\log_32}}}$, while  the memory-dependent bound is dominant for $M\geq \BO{\frac{n}{\nproc{}^{\log_32}}}$.

\section{Overview on algorithm strategy}\label{sec:overview}
In this section, we outline the main design principle shared by both our proposed integer multiplication algorithms. We will then delve in the details of \coim{} in Section~\ref{sec:algdef}, and of \cok{} in Section~\ref{sec:karatsuba}.
All the proposed algorithms follow a recursive strategy: sub-problems are recursively generated and opportunely assigned to the processors being used depending on the input size, the number of available processors, and the size of the local memory assigned to each processor. At the bottom of the recursion, sub-problems are assigned to single processors, which then compute their solution \emph{locally} without any further communication. 

Let us consider the recursion tree corresponding to the execution of the algorithm. Each node in the tree corresponds to an invocation of \coim{} (resp., \cok{}). The root of the tree corresponds to the initial invocation. The descendants of a node correspond to the recursive invocations of \coim{} (resp., \cok) used to compute the four (resp., three) sub-problems generated at each recursive level. At each recursion level, the algorithms may choose to schedule the execution of the recursive calls in two possible ways:
\begin{itemize}
    \item \textbf{Breadth First:} In a \emph{Breadth First Step} (BFS), the available $\nproc$ processors are divided into disjoint subsets of the same cardinality which are each assigned to compute \emph{in parallel} one of the recursive subproblems.
    \item \textbf{Depth First:} In a \emph{Depth First Step} (DFS), all processors are assigned to solve, each subproblem, \emph{in sequence} one at a time.
\end{itemize}

A BFS incurs a lower communication and computation cost than a DFS, as multiple recursive branches of the execution tree can be continued \emph{in parallel}. BF steps, however, require a higher available memory than DFS. Therefore, to minimize communication and computation cost, we pursue a strategy based on the opportune composition of BF and DF steps. Such a strategy is ultimately aimed to make the \emph{best possible use} of the available memory space. 

Let $\mathcal{P}$ denote the number of available processors. In each BFS, the number of processors assigned to each sub-problem is reduced by a factor equal to the number of sub-problems being generated (i.e., $4$ for \coim{}, and $3$ for $\cok{}$). Hence, after $\BO{\log \mathcal{P}}$ BFS, each of the generated sub-problem will be assigned to a single processor to be computed locally. Using DFS allows generating sub-problems of smaller input size whose result is computed 
using all the available processors. 

All our algorithms operate in two execution modes:
\begin{itemize}
    \item \textbf{Memory-independent execution mode (MI):} the algorithm executes $\BO{\log \mathcal{P}}$ consecutive BFS after which each of the sub-problems generated at the $\ell_{BF}$-th level is assigned to a \emph{single processor}, which computes the assigned sub-problem \emph{locally} with no further communication. This traversal scheme is only possible is the total combined available memory of the processors being used is sufficient given the input size (i.e., the number of digits in the base-$s$ expansion of the input integers). The denomination ``\emph{memory independent}'' highlights the fact that, provided that enough local memory space is available, the behavior of the algorithm while in the MI execution mode does not depend on the actual size of the local memories, which can then be assumed to be \emph{unlimited}, but rather only on the number of available processors.
    \item \textbf{Main execution mode:} given as input $n$-digit integers, in the main execution mode, the algorithm proceeds by executing  $\ell_{DF}$ consecutive DFS steps and then computes each of the sub-problems generated at the $\ell_{DF}$-th level in the MI execution mode. $\ell_{DF}$ is chosen as the minimum value be such that the size of the sub-problems generated at the $\ell_{DF}$-th level of recursion allows for them to be computed according to the MI scheme. As for both our algorithms the size of the sub-problems is reduced by half at every level of recursion, $\BO{\log_2 n/n_0}$ DF steps are executed in the main execution mode. This execution mode is also referred as the ``\emph{Limited Memory execution mode}'', as, contrary to the MI execution mode, the amount of available memory across the processors heavily affects the execution of the algorithm.
\end{itemize}

The values $\ell_{BF},\ell_{DF}$ and $n_0$ mentioned above depend on the specific algorithm considered, its memory requirement, the available computation resources (processors and memory), and the input size. In Section~\ref{sec:algdef} (resp., Section~\ref{sec:karatsuba}) we present the details of the proposed algorithms and their execution in both execution modes. Note that while it is possible to consider alternative strategies in which BFS and DFS are interleaved, we will show that our algorithms achieve optimal computation speedup and minimize the bandwidth communication cost.

\section{Parallel Algorithmic components}\label{sec:components}
In this section, we present subroutines for parallel addition, comparison, and subtraction of integer numbers, used in our algorithms for integer multiplications. 

\subsection{Parallel Sum with Distributed Memory}\label{sec:distsum}
Let $\mathbf{P}$ be a sequence of processors each equipped with a local memory of size $M$, and let $A$,$B$ be two $n$-digits integer numbers partitioned in $\mathbf{P}$ in $n/|\mathbf{P}|$ digits. That is 
    $$
    \begin{aligned}
A &= \left(A(\mathbf{P}[|\mathbf{P}|-1])\ldots A(\mathbf{P}[0])\right)_{s^{\frac{n}{|\mathbf{P}|}}},\label{eq:A11}\\
B &= \left(B(\mathbf{P}[|\mathbf{P}|-1])\ldots B(\mathbf{P}[0]\right)_{s^{\frac{n}{|\mathbf{P}|}}},\label{eq:B11}
\end{aligned}
$$
The parallel subroutine $\textsc{SUM}\left(\mathbf{P}, A,B, n/|\npr{P}|\right)$ computes $C=A+B$ in parallel, with $C$ being partitioned  in $\mathbf{P}$ in $n/|\mathbf{P}|$ digits, with $\mathbf{P}[|\mathbf{P}|-1]$ holding the most significant of the $n+1$ digits of $C$. Further, all processors in $\mathbf{P}$ hold a value $v\in\{0,1\}$ which denotes the most significant digit of $C$. In the following, we assume $n$ and $|\mathbf{P}|$ to be integer powers of two. If that is not the case, our algorithm may be applied with minor adjustments (e.g., padding).\\

The algorithm follows a recursive strategy. 
When $\textsc{SUM}\left(\mathbf{P}, A,B,, n/|\mathbf{P}|\right)$ is invoked, if  $|\mathbf{P}|=1$, the single processor $\mathbf{P}[0]$ computes $C$ and $v$ locally. If $|\mathbf{P}|>1$, the algorithm executes the following operations: 
\begin{enumerate}
    \item The sequence of available processors is divided in subsequences:
    \begin{equation}\label{eq:processorpartition2}
    \begin{aligned}
    \mathbf{P}' &= [\mathbf{P}[(|\mathbf{P}|/2)-1],\ldots,\mathbf{P}[0]]\\
    \mathbf{P}'' &= [\mathbf{P}[(|\mathbf{P}|-1],\ldots,\mathbf{P}[|\mathbf{P}|/2]]\\
\end{aligned}
\end{equation}
Correspondingly, let:
\begin{equation}\label{eq:partitionAB2}
        \begin{aligned}
    A_0 &= \left(A(\mathbf{P}[|\mathbf{P}|/2-1])\ldots A(\mathbf{P}[0])\right)_{s^{\frac{s}{|\mathbf{P}|}}}\\
    A_1 &= \left(A(\mathbf{P}[|\mathbf{P}|-1])\ldots A(\mathbf{P}[|\mathbf{P}|/2])\right)_{s^{\frac{s}{|\mathbf{P}|}}}\\
    B_0 &= \left(B(\mathbf{P}[|\mathbf{P}|/2-1])\ldots B(\mathbf{P}[0])\right)_{s^{\frac{s}{|\mathbf{P}|}}}\\
    B_1 &= \left(B(\mathbf{P}[|\mathbf{P}|-1])\ldots B(\mathbf{P}[|\mathbf{P}|/2])\right)_{s^{\frac{s}{|\mathbf{P}|}}}\\
\end{aligned}
\end{equation}
where $A_0$ and $B_0$ (resp., $A_1$ and $B_1$) are partitioned in $\mathbf{P}'$ (resp. $\mathbf{P}''$) in $n/|\mathbf{P}|$ digits. 
\item $\textsc{SUM}$ then invokes $\textsc{SUM}(\mathbf{P'}, A_0,B_0, n/|\mathbf{P}|)$ and the \emph{auxiliary subroutine} $\textsc{SUMA}(\mathbf{P'}, A_1,B_1, n/|\mathbf{P}|)$ to be executed in parallel.
$\textsc{SUMA}\left(\mathbf{P''}, \{A',B'\}, n\right)$ computes
\begin{equation*}
    \begin{aligned}
        C''_0 &= (A'+B')\mod s^{n|\mathbf{P'}|/|\mathbf{P}|}\\
        u_0 &=   \lfloor (A''+B'')/ s^{n|\mathbf{P'}|/|\mathbf{P}|}\rfloor\\
        C''_1 &= (A'+B'+1)\mod s^{n|\mathbf{P'}|/|\mathbf{P}|}\\
        u_1 &=   \lfloor (A''+B''+1)/ s^{n|\mathbf{P'}|/|\mathbf{P}|}\rfloor\\
    \end{aligned}
\end{equation*}
    That is, $C''_i$ (resp., $u_i$), for $i\in\{0,1\}$, denotes the value corresponding to the $n|\mathbf{P''}|/|\mathbf{P}|$ least significant digits (resp., the most significant digit) of the sum $A'+B'+i$. Such values are used to \emph{speculatively precalculate} the value of the digits of the final output $C$ to be partitioned in $\mathbf{P''}$ for the two possible values of the carryover (i.e., $v$) of the sum $A_0+B_0$ computed by $P_0$.  All the computed $C_i''$ are partitioned in $\mathbf{P''}$ in $n/|\mathbf{P}|$ digits, and each processor in $\mathbf{P''}$ holds a copy of the $u_i$'s.
\item Let us denote as $C' = (A_0+B_0)\mod s^{n/2}$ and $v' = \lfloor (A_0+B_0)/ s^{n/2}\rfloor$ the output values of the subroutine $\textsc{SUM}(\mathbf{P'}, A_0,B_0, n/|\mathbf{P}|)$. Clearly, $C'$ corresponds the the value of the $n/2$ least significant digits of $A+B$. In parallel, each processor $\mathbf{P}'[j]$, for $0\leq j\leq |\mathbf{P}|/2-1$, sends to $\mathbf{P}''[j]$ the value $v'$, and then removes it from its local memory. 
\item Upon receipt each $\mathbf{P}''[j]$ assigns $C(\mathbf{P''}[j]) = C''_{v'}(\mathbf{P''}[j])$ and $v = u_{v'}$.
In parallel, each $\mathbf{P}'[j]$ assigns $C(\mathbf{P'}[j]) = C'(\mathbf{P'}[j])$.
At the end of this step, we have:
\begin{equation*}
\begin{split}
    C = &(v, C(\mathbf{P''}[|\mathbf{P}|/2-1]]), C(\mathbf{P''}[|\mathbf{P}|/2-2]),\ldots, C(\mathbf{P''}[0]), C(\mathbf{P'}[|\mathbf{P}|/2-1]),\\ & C(\mathbf{P'}[|\mathbf{P}|/2-2]),\ldots, C(\mathbf{P'}[0]))_{s^{n/|\mathbf{P}|}}.
\end{split}
\end{equation*}
That is, $(C\mod s^n)$ is partitioned in $\mathbf{P}$ in $n/|\mathbf{P}|$ digits, and all processors in $\mathbf{P''}$ have information on the most significant digit of $C$. Then, each $\mathbf{P}'[j]$ (resp., $\mathbf{P}''[j]$) removes from its memory the temporary value $C'(\mathbf{P'}[j])$ (resp., $C''_0(\mathbf{P''}[j])$, $C''_1(\mathbf{P''}[j])$, $u_0$ and $u_1$).
\item  In parallel, each processor $\mathbf{P}''[j]$, sends to $\mathbf{P}'[j]$ a copy of $v$.
\end{enumerate}

After the completion of the initial invocation of $\textsc{SUM}$, it easily possible to reconstruct the full $C=A+B$ by having $\mathbf{P}[|\mathbf{P}|-1]$ append $v$ as the most significant digit of $C(\mathbf{P}[|\mathbf{P}|-1])$. Once $C$ is computed, all processors in $\mathbf{P}$ may remove $v$ from their local cache.

To complete the description of the algorithm, we present the details of the auxiliary procedure $\textsc{SUMA}$: If $|\mathbf{P}|=1$, $\textsc{SUMA}\left(\mathbf{P}, \{A,B\}, n'\right)$ computes locally the following values:

\noindent\begin{minipage}{.5\linewidth}
\begin{align*}
        C_0 &= (A+B)\mod s^{n'}\\
        u_0 &=   \lfloor (A+B)/ s^{n'}\rfloor
    \end{align*}
\end{minipage}%
\begin{minipage}{.5\linewidth}
\begin{align*}
        C_1 &= (A+B+1)\mod s^{n'}\\
        u_1 &=   \lfloor (A+B+1)/ s^{n'}\rfloor
    \end{align*}
\end{minipage}

\noindent where $n'$ denotes the number of digits of $A$ and $B$ stored in the local memory of each processor.

\noindent Instead, if $|\mathbf{P}|>1$, the algorithm executes the following operations:
\begin{enumerate}
    \item  As done in point (2) of the description of $\textsc{ADD}$, the sequence of available processors is divided in the two subsequences $\mathbf{P'}$ and $\mathbf{P''}$~\eqref{eq:processorpartition2}, and the input $A$ (resp., $B$) is partitioned in $A_0$ and $A_1$ (resp., $A_1$ and $B_1$)~\eqref{eq:partitionAB2}. 
\item $\textsc{SUMA}$ then recursively invokes $\textsc{SUMA}(\mathbf{P'}, A_0,B_0, n')$ and $\textsc{SUMA}(\mathbf{P''}, A_1,B_1, n')$ to be executed in parallel. In the following we denote as $C'_0$, $C'_1$, $u'_0$ and $u'_1$ (resp., $C''_0$, $C''_1$, $u''_0$ and $u''_1$) the output values of the former (resp., latter) call.
\item In parallel, each processor $\mathbf{P}'[j]$, for $j=0,1,\ldots, |\mathbf{P}|/2-1$, sends to $\mathbf{P}''[j]$ the values $u'_0$ and $u'_1$, and then removes them from its local memory. Upon receipt each $\mathbf{P}''[j]$ assigns

\noindent\begin{minipage}{.5\linewidth}
\begin{align*}
        C_0(\mathbf{P''[j]}) &= C''_{u'_0}(\mathbf{P''[j]})\\
        u_0 &=  u''_{u'_0}
\end{align*}
\end{minipage}%
\begin{minipage}{.5\linewidth}
\begin{align*}
        C_1(\mathbf{P''[j]}) &= C''_{u'_1}(\mathbf{P''[j]})\\
        u_1 &=  u''_{u'_1}
\end{align*}
\end{minipage}
\vspace{2mm}

In parallel, each $\mathbf{P}'[j]$ assigns  $C_0(\mathbf{P'[j]}) = C'_{0}(\mathbf{P'[j]})$ and $C_1(\mathbf{P'[j]}) = C'_{1}(\mathbf{P'[j]})$. Then  each  $\mathbf{P'}[j]$ (resp., $\mathbf{P''}$[j]) removes the values $C'_i$, $u'_i$ (resp., $C''_i$, $u'_i$) from its local memory, for $i=0,1$.
\item  In parallel, each processor $\mathbf{P}''[j]$, sends to $\mathbf{P}'[j]$ a copy of $u_0$ and $u_1$. 
\end{enumerate}

The following lemma characterizes the memory requirement, the computational time and \io{} cost of $\textsc{SUM}$:
\begin{lemma}\label{lem:costsum}
Let $A$,$B$ be two $n$-digit integers initially partitioned in a sequence of processors $\mathbf{P}$ in $n/|\mathbf{P}|$ digits. If each processor in $\mathbf{P}$ is equipped with a local memory of size $4(n/|\mathbf{P}|+1)$, algorithm $\textsc{SUM}$ computes the sum $C=A+B$. We have:
\begin{align*}
T_\textsc{SUM}\left(n,|\mathbf{P}|\right)&\leq 6n/|\mathbf{P}|+4\log_2 |\mathbf{P}|\\ 
B_\textsc{SUM}\left(n,|\mathbf{P}|\right)&\leq 4\log_2 |\mathbf{P}|\\ 
L_\textsc{SUM}\left(n,|\mathbf{P}|\right)&\leq 2\log_2 |\mathbf{P}|
\end{align*}
\end{lemma}
\begin{proof}
\textsc{SUM} correctly computes $C=A+B$ by inspection. We focus on the analysis of the computational and \io{} requirement of the auxiliary subroutine $\textsc{SUMA}$, as this subsumes the analysis of $\textsc{SUM}$.
At any point during the computation each processor must maintain in memory at most $n/|\mathbf{P}|$ digits of the input integers $A$ and $B$, $n/|\mathbf{P}|$ digits for each the values $C_0$ and $C_1$ obtained from the last recursive call, two values $u_0$ and $u_1$ returned by the previous recursive call, and at most two copies of $u'_0$ and $u'_1$ received from another processor and not yet deleted. Hence, $\textsc{SUM}$ can be executed if each processor is equipped with a local memory of size $4(n/|\mathbf{P}|+1)$.
For $|\mathbf{P}|>1$, $\textsc{SUMA}$ recursively invokes two instances to be executed in parallel. In step (3) and (4) of the algorithm, half of the processors send two memory words to distinct processors in the remaining half. Finally, at step (3), up to four comparisons are necessary to assign the values $C_0$,$C_1$,$u_0$ and $u_1$. Hence, we have that the computational cost, the bandwidth cost and the latency of $\textsc{SUMA}$ satisfy:
\begin{equation}\label{eq:suma1}
    \begin{aligned}
    T_{\textsc{SUMA}}\left(n,|\mathbf{P}|\right) &\leq T_{\textsc{SUMA}}\left(n/2,|\mathbf{P}|/2\right) + 4\\
    B_{\textsc{SUMA}}\left(n,|\mathbf{P}|\right) &\leq B_{\textsc{SUMA}}\left(n/2,|\mathbf{P}|/2\right) + 4\\
    L_{\textsc{SUMA}}\left(n,|\mathbf{P}|\right) &\leq L_{\textsc{SUMA}}\left(n/2,|\mathbf{P}|/2\right) + 2
    \end{aligned}
\end{equation}
In the base case, for $|\mathbf{P}|=1$, $\textsc{SUMA}$ computes the sums $C_0$ and $C_1$ locally without any further communication (i.e., $B_{\textsc{SUMA}}\left(n/|\mathbf{P}|,1\right)= L_{\textsc{SUMA}}\left(n/|\mathbf{P}|,1\right) = 0$). As the numbers being added have at most $n/|\mathbf{P}|$ digits, each value can be computed using at most $3n/|\mathbf{P}|$ elementary operations (i.e.,  $T_{\textsc{SUMA}}\left(n/|\mathbf{P}|,1, M\right)<6n$). Thus, from~\eqref{eq:suma1} we have  
\begin{equation*}
    \begin{aligned}
    T_{\textsc{SUMA}}\left(n,|\mathbf{P}|\right) &\leq T_{\textsc{SUMA}}\left(n/|\mathbf{P}|,1\right) + \sum_{i=1}^{\log_2 |\mathbf{P}|}4\\
    B_{\textsc{SUMA}}\left(n,|\mathbf{P}|\right) &\leq B_{\textsc{SUMA}}\left(n/|\mathbf{P}|,1 \right) + \sum_{i=1}^{\log_2 |\mathbf{P}|}4\\
    L_{\textsc{SUMA}}\left(n,|\mathbf{P}|\right) &\leq L_{\textsc{SUMA}}\left(n/|\mathbf{P}|,1 \right) + \sum_{i=1}^{\log_2 |\mathbf{P}|}2
    \end{aligned}
\end{equation*}
The lemma follows
\end{proof}
When summing $n$-digits integers using $\npr{P}$ processors, such that $\npr{P}\log_2\npr{P}\in \BO{n}$, $\textsc{SUM}$ achieves optimal speedup. 
While the presentation discussed here focuses on the sum of two integers, the procedure can be easily extended to more addends. The computation and \io{} cost scales \emph{linearly} with the number of addends.

\subsection{Parallel Comparison with Distributed Memory}
In this subsection we describe how given two $n$-digit input integers $A, B$ partitioned in a sequence of processors $\mathbf{P}$ in $n/|\mathbf{P}|$ digits, is it possible to efficiently determinate whether $A\geq B$. Our algorithm $\textsc{COMPARE}\left(\mathbf{P}, A,B\right)$ achieves asymptotically computational speedup for $n\geq\Omega\left( |\mathbf{P}|\log_2|\mathbf{P}|\right)$.\\

At the end of an invocation of $\textsc{COMPARE}\left(\mathbf{P},A,B\right)$ each processor in $\mathbf{P}$ holds a flag $f$ such that $f=0$ if $A=B$, $f=1$ if $A>B$, and $f=-1$ if $B>A$. The algorithms employs a recursive strategy. If $|\mathbf{P}|=1$, the single processor in $\mathbf{P}$ computes the value of the flag $f$ locally. If that is not the case, the following operations are executed:
\begin{enumerate}
     \item  Divide the sequence of available processors in the two subsequences $\mathbf{P'}$ and $\mathbf{P''}$ as  in~\eqref{eq:processorpartition2}, and the input $A$ (resp., $B$)  in $A_0$ and $A_1$ (resp., $A_1$ and $B_1$) as in~\eqref{eq:partitionAB2}. 
     \item Recursively invoke $\textsc{COMPARE}(\mathbf{P'}, A_0,B_0)$ and $\textsc{COMPARE}(\mathbf{P''}, A_1,B_1)$ to be executed in parallel.
     \item \sloppy Let $f'$ (resp., $f''$) denote the flag computed by $\textsc{COMPARE}(\mathbf{P'}, A_0,B_0)$  (resp., $\textsc{COMPARE}(\mathbf{P''}, A_1,B_1)$). In parallel, each processor $\mathbf{P}'[i]$, for $0\leq i\leq |\mathbf{P}|/2-1$, sends to $\mathbf{P}''[i]$ the flag $f'$ and then removes them from its cache. 
     \item Upon receipt, the $\mathbf{P}''[i]$'s, compute the new flag $f$:
     \begin{equation*}
         f = \begin{cases}f' &if\ f'\neq 0\\
         f'' &if\ f'=0\\\end{cases},
     \end{equation*}
     and then remove $f'$ and $f''$ from their local cache.
     \item In parallel, each processor $\mathbf{P}''[i]$, for $0\leq i\leq |\mathbf{P}|/2-1$, sends to $\mathbf{P}'[i]$ a copy of the flag $f$.
\end{enumerate}

\begin{lemma}\label{lem:costcompare}
Using algorithm $\textsc{COMPARE}$ to compare two $n$-digits integers $A$ and $B$ using $|\mathbf{P}|$ processors requires each processor being equipped with a local memory of size $2n/|\mathbf{P}|+2$. Further:
\begin{align*}
    T_\textsc{COMPARE}\left(n,|\mathbf{P}|\right) &\leq \frac{n}{|\mathbf{P}|} + \log_2 |\mathbf{P}|,\\
    B_\textsc{COMPARE}\left(n,|\mathbf{P}|\right) &\leq \log_2 |\mathbf{P}|,\\
    L_\textsc{COMPARE}\left(n,|\mathbf{P}|\right) &\leq \log_2 |\mathbf{P}|.
\end{align*}
\end{lemma}
\begin{proof}
At any time during the computation each processor needs to maintain in its local cache at most $n/|\mathbf{P}|$ digits of $A$ and $B$, the value of the flag $f$ returned by the last invocation of $\textsc{COMPARE}$, and the value of a second flag received by another processor. Hence the memory requirement is bounded by $2(n/|\mathbf{P}|+1)$.
The analysis of the computation and \io{} cost is analogous to that discussed in the proof of Lemma~\ref{lem:costsum} for \textsc{SUM}.
\end{proof}

\subsection{Parallel Difference with Distributed Memory}
Let $A$ $B$ be two positive integer numbers whose base-$s$ expansion has at most $n$ digits, and assume them to be initially partitioned in $\mathbf{P}$ in $n/|\mathbf{P}|$ digits. The parallel subroutine $\textsc{DIFF}\left(\mathbf{P}, A,B, n/|\mathbf{P}|\right)$ computes $C=|A-B|$, with $C$ being partitioned  in $\mathbf{P}$ in $n/|\mathbf{P}|$ digits, and a flag $f$ such that $f=0$ if $A=B$, $f=1$ if $A>B$, and $f=-1$ if $B>A$. In the following, we assume $n$ and $|\mathbf{P}|$ to be integer powers of two. If that is not the case, our algorithm may be applied with  minor adjustments (e.g., padding).\\

 When $\textsc{DIFF}\left(\mathbf{P}, A,B, n/|\mathbf{P}|\right)$ is invoked, if $|\mathbf{P}|=1$ the single processor $\mathbf{P}[0]$, computes $|A-B|$ and the value $f$ locally. If $|\mathbf{P}|>1$, the following operations are executed:
\begin{enumerate}
    \item $\textsc{COMPARE}\left(\mathbf{P}, A,B\right)$ is invoked to set the value $f$. If $f=0$ then each processor $\mathbf{P[i]}$ sets $C(\mathbf{P}[i])=0$ and no further operation is executed. Instead, if $f=1$ (resp., $f=-1$), the algorithm computes $A-B$ (resp., $B-A$) using a recursive divide and conquer approach. For the sake of simplicity, in the following we assume $A>B$. The case $B>A$ follows by swapping $A$ and $B$ in the following. 
    \item Let $\mathbf{P}'$ and $\mathbf{P}''$ be defined as in~\eqref{eq:processorpartition2}, and let $A_0$, $A_1$, $B_0$ and $B_1$ defined as in~\eqref{eq:partitionAB2}. $\textsc{DIFF}$ proceeds by invoking, in parallel, two, slightly different, recursive subroutines:
    \begin{itemize}
    \item $\textsc{DIFFL}\left(\mathbf{P'}, A_0,B_0, n/|\mathbf{P}|\right)$ computes $C'= (A+s^{n|\mathbf{P'}|/|\mathbf{P}|}-B)\mod s^{n|\mathbf{P'}|/|\mathbf{P}|}$, and partitions it in $\mathbf{P'}$ in $n|\mathbf{P'}|/|\mathbf{P}|$ digits. $C'$ corresponds to the $n/2$ least significant digits of $|A-B|$. Further, $\textsc{DIFFL}$ computes  $b'$ such that $b'=0$ if $A\geq B$ and $b'=1$ otherwise.
    $b'$ is used to  denote whether when computing $A_0-B_0$ it will be necessary to ``\emph{borrow}'' from the $(n/2)$-th digit of $A$(i.e,, if yes $b'=1$, if no $b'=0$). At the end of $\textsc{DIFFL}$, all processors in $\mathbf{P}$ have a copy of $b'$ in their local cache.
    \item $\textsc{DIFFR}\left(\mathbf{P''}, A_1,B_1, n/|\mathbf{P}|\right)$ computes
    \begin{align*}
        C''_0 &= \left(A_1+s^{n|\mathbf{P''}|/|\mathbf{P}|}-B_1\right)\mod s^{n|\mathbf{P}''|/|\mathbf{P}|}\\
        b''_0 &=   \mathbf{1}(A_1 \geq B_1)\\
        C''_{1} &= \left(A_1+s^{n|\mathbf{P''}|/|\mathbf{P}|}-B_1-1\right)\mod s^{n|\mathbf{P''}|/|\mathbf{P}|}\\
        b''_1 &=   \mathbf{1}(A_1 -1 \geq B_1)
    \end{align*}
    The values $C''_i$'s are used to \emph{speculatively calculate} the $n/2$ most significant digits of $A-B$ depending whether it will necessary to \emph{borrow} from $A_1$ when computing $A_0-B_0$ (i.e., $i=0$) or not (i.e, $i=1$). Similarly, $f''_i$'s are used to \emph{speculatively calculate} whether when computing the difference $A_1-B_1-i$ it will be necessary to borrow, depending on whether it will necessary to \emph{borrow}. The computed $C_i''$'s are partitioned in $\mathbf{P''}$ in $n/|\mathbf{P}|$ digits, and each processor in $\mathbf{P''}$ holds a copy of the $b''_i$'s.
\end{itemize}
    \item Each processor $\mathbf{P'}[j]$, for $j=0,1,\ldots, |\mathbf{P}|/2-1$, sends to $\mathbf{P''}[j]$, the value $b'$ and then removes it from its cache. Upon receipt, each $\mathbf{P''}[j]$ assigns $C[\mathbf{P''}[j]] = C''_{b'}[\mathbf{P''}[j]]$, and then removes $b'$ $C''_{0}[\mathbf{P''}[j]]$, $C''_1[\mathbf{P''}[j]]$, $b''_0$, and $b''_1$ from its local memory. In parallel, each $\mathbf{P'}[j]$ assigns $C[\mathbf{P'}[j]] = C[\mathbf{P'}[j]]$, and removes $C[\mathbf{P'}[j]]$ from its local memory.
\end{enumerate}
As desired, the value $C=|A-B|$ is partitioned in $\mathbf{P}$ in $n/|\mathbf{P}|$ digits.

To conclude the description of $\textsc{DIFF}$, we now present the details of the recursive subroutines $\textsc{DIFFL}$ and $\textsc{DIFFR}$:

When $\textsc{DIFFL}\left(\mathbf{P}, \{A,B\}, n/|\mathbf{P}|\right)$ is invoked, if $|\mathbf{P}|=1$ the single processor $\mathbf{P}[0]$, computes $|A-B|$ and the value $b$ locally. If $|\mathbf{P}|>1$, the following operations are executed:
\begin{enumerate}
    \item The sequence of available processors is divided in the two subsequences $\mathbf{P'}$ and $\mathbf{P''}$~\eqref{eq:processorpartition2}, and the input $A$ (resp., $B$) is partitioned in $A_0$ and $A_1$ (resp., $A_1$ and $B_1$)~\eqref{eq:partitionAB2}. 
    \item $\textsc{DIFFL}$ recursively invokes $\textsc{DIFFL}\left(\mathbf{P}', A_0,B_0,  n/|\mathbf{P}|\right)$ and $\textsc{DIFFR}\left(\mathbf{P}'', A_1,B_1,  n/|\mathbf{P}|\right)$ to be executed in parallel. Let $C'$ and $b'$ (resp., $C''_0$,$C''_1$,$b''_0$, $b''_1$) denote the output of the former (resp., latter) call.
    \item $\textsc{DIFFL}$ proceeds to compute $|A-B|$ following the operations discussed in step (3) of the main procedure $\textsc{DIFF}$. 
    \item Additionally, when  the processor  $\mathbf{P''}[j]$'s receive $b'$ from $\mathbf{P'}[j]$, it assigns $b= b''_{b'}$, removes $b'$, $b''_0$ and $b''_1$ from its local memory, and then sends a copy of $b$ to $\mathbf{P'}[j]$.
\end{enumerate}

Finally, when $\textsc{DIFFR}\left(\mathbf{P}, A,B, n/|\mathbf{P}|\right)$ is invoked, if $|\mathbf{P}|=1$ the single processor $\mathbf{P}[0]$, computes $C_0=|A-B|$, $C_1 = |A-B-1|$, $b_0=\mathbf{1}(A\geq B)$ and $b_1=\mathbf{1}(A-1\geq B)$ locally. If $|\mathbf{P}|>1$, the following operations are executed:
\begin{enumerate}
    \item The sequence of available processors is divided in the two subsequences $\mathbf{P'}$ and $\mathbf{P''}$~\eqref{eq:processorpartition2}, and the input $A$ (resp., $B$) is partitioned in $A_0$ and $A_1$ (resp., $A_1$ and $B_1$)~\eqref{eq:partitionAB2}. 
    \item $\textsc{DIFFR}$ recursively invokes $\textsc{DIFFR}\left(\mathbf{P}', A_0,B_0,  n/|\mathbf{P}|\right)$ and $\textsc{DIFFR}\left(\mathbf{P}'', A_1,B_1,  n/|\mathbf{P}|\right)$ to be executed in parallel. Let $C_0'$,$C_1'$, $b'_0$ and $b'_1$ (resp., $C''_0$,$C''_1$,$b''_0$, $b''_1$) denote the output of the former (resp., latter) call.
    \item Each processor $\mathbf{P}'[j]$, for $i=0,1,\ldots, |\mathbf{P}|/2-1$, sends to $\mathbf{P}''[j]$, the values $b'_0$ and $b'_1$, and then removes them from its cache. 
    \item Upon receipt, each $\mathbf{P}''[i]$ assigns $C_0[\mathbf{P''}[j]] = C''_{b'_0}[\mathbf{P''}[j]]$ (resp., $C_1[\mathbf{P''}[j]] = C''_{b'_1}[\mathbf{P''}[j]]$). Further, it assigns  $b_0 = b''_{b'_0}$ (resp., $b_1 = b''_{b'_1}$), and then removes $C''_0$,$C''_1$, $b'_0$, $b'_1$, $b''_0$ and $b''_1$ from its local memory. In parallel, each $\mathbf{P'}[j]$  assigns $C_0(\mathbf{P'}[j]) = C'_{0}(\mathbf{P'}[j])$ (resp., $C_1(\mathbf{P'}[j]) = C'_{1}(\mathbf{P'}[j])$), and then removes $C''_0$an $C''_1$ from its local memory
    \item In parallel, each $\mathbf{P}''_[j]$ sends to $\mathbf{P}'_[j]$ a copy of $b_0$ and $b_1$.
\end{enumerate}
At the end of these operations, the values $C_0$, $C_1$ are partitioned in $\mathbf{P}$ as desired.

The following lemma characterizes the memory requirement, the computational time and \io{} cost of $\textsc{DIFF}$:
\begin{lemma}\label{lem:costdiff}
Using algorithm $\textsc{DIFF}$ to compute the difference of two $n$-digit integers $A-B$ using a sequence of processors $\mathbf{P}$, requires each processor to be equipped with a local memory of size at least $4n/|\mathbf{P}|+5$. We have:
\begin{align*}
T_\textsc{DIFF}\left(n,|\mathbf{P}|\right)&\leq 7n/|\mathbf{P}|+5\log_2 |\mathbf{P}|\\ 
B_\textsc{DIFF}\left(n,|\mathbf{P}|\right)&\leq 5\log_2 |\mathbf{P}|\\ 
L_\textsc{DIFF}\left(n,|\mathbf{P}|\right)&\leq 3\log_2 |\mathbf{P}|
\end{align*}
\end{lemma}
\begin{proof} \textsc{DIFF} correctly computes $|A-B|$ and the sign of $A-B$ by inspection.

By Lemma~\ref{lem:costcompare}, the initial invocation of $\textsc{COMPARE}\left(\mathbf{P}, A,B\right)$, requires each processor in $\mathbf{P}$ to be equipped with a memory of size at least $2(n/|\mathbf{P}|+1)$. For $j=0,1,\ldots |\mathbf{P}|-1$, the locations used to store the digits of $A(\mathbf{P}[j])$ and $B(\mathbf{P}[j])$ are reused in the remainder of the computation. At any point during the computation each processor must maintain in memory at most $n/|\mathbf{P}|$ digits of the input integers $A(\mathbf{P}[j])$ and $B(\mathbf{P}[j])$ , $2n/|\mathbf{P}|$ (resp., $2$) digits for the values $C_0$ and  $C_1$ (resp., $b_0$ and $b_1$) returned by the last recursive call, at most two copies of $b'_0$ and $b'_1$ received from another processor and not yet deleted, and the flag $f$ computed by the initial call to $\textsc{COMPARE}$. Hence, $\textsc{DIFF}$ can be executed if each processor is equipped with a local memory of size $4n/|\mathbf{P}|+5$.
We focus on the analysis of the computational and \io{} requirement of the auxiliary subroutine $\textsc{DIFFR}$, as this subsumes the analysis of $\textsc{DIFFL}$. 
For $|\mathbf{P}|>1$, $\textsc{DIFFR}$ recursively invokes two instances to be executed in parallel. In step (3) and (5) of the algorithm, half of processors send two memory words to distinct processors in the remaining  half. Finally, at step (4), up to four comparisons are necessary to assign the values $C_0$,$C_1$,$b_0$ and $b_1$. Hence, we have that the computational cost, the bandwidth cost and the latency of $\textsc{DIFFR}$ satisfy:
\begin{equation}\label{eq:diff1}
    \begin{aligned}
    T_{\textsc{DIFFR}}\left(n,|\mathbf{P}|\right) &\leq T_{\textsc{DIFFR}}\left(n/2,|\mathbf{P}|/2\right) + 4\\
    B_{\textsc{DIFFR}}\left(n,|\mathbf{P}|\right) &\leq B_{\textsc{DIFFR}}\left(n/2,|\mathbf{P}|/2\right) + 4\\
    L_{\textsc{DIFFR}}\left(n,|\mathbf{P}|\right) &\leq L_{\textsc{DIFFR}}\left(n/2,|\mathbf{P}|/2\right) + 2
    \end{aligned}
\end{equation}
In the base case, for $|\mathbf{P}|=1$, $\textsc{DIFFR}$ computes the differences $C_0$ and $C_1$ locally without any further communication (i.e., $B_{\textsc{SUMA}}\left(n/|\mathbf{P}|,1,M\right)= B_{\textsc{DIFFR}}\left(n/|\mathbf{P}|,1,M\right) = 0$). As the numbers being added have at most $n/|\mathbf{P}|$ digits, each value can be computed using at most $3n/|\mathbf{P}|$ elementary operations (i.e.,  $T_{\textsc{DIFFR}}\left(n/|\mathbf{P}|,1, M\right)<6n$). Thus, from~\eqref{eq:diff1} we have  
\begin{equation*}
    \begin{aligned}
    T_{\textsc{DIFFR}}\left(n,|\mathbf{P}|\right) &\leq T_{\textsc{DIFFR}}\left(n/|\mathbf{P}|,1\right) + \sum_{i=1}^{\log_2 |\mathbf{P}|-1}4\\
    B_{\textsc{DIFFR}}\left(n,|\mathbf{P}|\right) &\leq B_{\textsc{DIFFR}}\left(n/|\mathbf{P}|,1 \right) + \sum_{i=1}^{\log_2 |\mathbf{P}|}4\\
    L_{\textsc{DIFFR}}\left(n,|\mathbf{P}|\right) &\leq L_{\textsc{DIFFR}}\left(n/|\mathbf{P}|,1 \right) + \sum_{i=1}^{\log_2 |\mathbf{P}|-1}2
    \end{aligned}
\end{equation*}
The lemma follows by summing the computation time (resp., \io{} cost) of $\textsc{DIFFR}$ with that of the initial invocation of $\textsc{COMPARE}$ (according to Lemma~\ref{lem:costcompare}).
\end{proof}
When computing the difference of $n$-digits integers using $\npr{P}$ processors, such that $\npr{P}\log_2\npr{P}\in \BO{n}$, $\textsc{DIFF}$ achieves optimal speedup. 
As for $\textsc{SUM}$, while the presentation discussed here focuses on the difference of two integers, the procedure can be easily extended to more inputs. The computation and \io{} cost scales \emph{linearly} with the number of input integers.

\section{\coim{}: Communication-optimal Parallel Standard Integer Multiplication}\label{sec:algdef}
In this section we present our first algorithm, \coim{} (Communication Optimal Parallel Standard Integer Multiplication), a parallel implementation of the sequential recursive \emph{long multiplication} algorithm, which given two input $n$-digit integers computes their product by executing $\BO{n^2}$ operations. The sequential algorithm, henceforth referred to as $\textsc{SLIM}$ (Sequential Long Integer Multiplication), follows a simple recursive strategy: If $n=1$, the product $C=A\times B$ is computed directly. Otherwise, let:

\noindent\begin{minipage}{.5\linewidth}
\begin{align*}
        A_0 &= \left(A[\lceil n/2\rceil-1],\ldots, A[0]\right)_s\\  
B_0 &= \left(B[\lceil n/2\rceil-1],\ldots, B[0]\right)_s
\end{align*}
\end{minipage}%
\begin{minipage}{.5\linewidth}
\begin{align*}
        A_1 &= \left(A[n-1],\ldots, A[\lceil n/2\rceil]\right)_s\\  
        B_1 &= \left(B[n-1],\ldots, B[\lceil n/2\rceil]\right)_s
\end{align*}
\end{minipage}
\vspace{2mm}

\noindent $\textsc{SLIM}$ is then recursively invoked to compute the products $C_0=A_0\times B_0$, $C_1 = A_0\times B_1$, $C_2= A_1\times B_0$ and $C_3=A_1\times B_1$. Finally, $C$ is computed as $C=C_0 + s^{n/4}(C_1+C_2)+s^{n/2}C_3$. 

\begin{fact}\label{lem:seqcostlim}
Algorithm $\textsc{SLIM}$ computes the product of two $n$-digit integers using at most $8n^2$ digit wise operations and memory space of size at most $8n$. 
\end{fact}

Our parallel algorithm $\coim{}$ implements the recursive scheme of $\textsc{SLIM}$ to take advantage of multiple available processors following the general outline discussed in Section~\ref{sec:overview}. We assume that the input factor $n$-digit integers to be partitioned in a sequence of processors $\mathbf{P}$ in $n/|\mathbf{P}|$ digits at the beginning of the computation. Further, we assume each processor to be equipped with a local memory of size $M\geq 80n/|\mathbf{P}|$. Finally, to simplify our analysis, we assume $n$ to be an integer power of two, and $|\mathbf{P}|$ to be an integer power of four, with $n\geq |\mathbf{P}|$. If that is not the case, the input can be padded with dummy digits, and/or some of the available processors may not be used. The following description and analysis remain correct in these cases with small corrections of the constant factors.

If $n\leq M\sqrt{|\mathbf{P}|}/12$, then the product is computed by \coim{} in the \emph{MI execution mode} (i.e.,$\coim{}_{MI}$) in $\log_4 |\mathbf{P}|$ \emph{breadth-first} recursive steps. If that is not the case, \coim{} proceeds in the \emph{main execution mode} by executing up to $\lceil \log_2 n/(24M\sqrt{|\mathbf{P}|})\rceil$ \emph{depth-first} recursive steps. The sub-problems generated after such steps will have input size which, given the size of the available memory,  allows for their solution to be calculated in the MI execution mode. \coim{} uses the sequential algorithm $\textsc{SLIM}$ to compute the products of integers $\emph{locally} $. Clearly, any sequential algorithm can be used in place of it. In the following, we present in detail and analyze \coim{}'s execution in the MI and main execution mode.

\subsection{\coim{} in the MI execution mode}
We denote the operations of $\coim{}$ in the MI execution mode as $\coim{}_{MI}$. 
Consider an invocation of $\coim{}_{MI}$  to multiply two $n$-digits integers using a sequence of processors $\mathbf{P}$. If $|\mathbf{P}| = 1$, then the product $C=A\times B$ is computed by the single available processor using the sequential algorithm $\textsc{SLIM}$.
If $|\mathbf{P}| > 1$, $\coim{}_{MI}$ proceeds as follows:
\begin{enumerate}
    \item \textbf{Splitting:} Let $\mathbf{P}'$ and $\mathbf{P}''$ be two subsequences of $\mathbf{P}$ defined as in~\eqref{eq:processorpartition2}, and let $A_0$ and $A_1$ ($B_0$ and $B_1$) be defined as~\ref{eq:partitionAB2}. $A_0$ (resp., $B_0$) is partitioned in $\mathbf{P}'$ and $A_1$ (resp., $B_1$) is partitioned in $\mathbf{P}'$ in $n/|\mathbf{P}|$ digits.
 \coim{} divides the available $\mathbf{P}$ processors into four sub-sequences each with $|\mathbf{P}|/4$ processors:
$$
\begin{aligned}
    \mathbf{P_0} &= [\mathbf{P}[(|\mathbf{P}|/2)-2],\ldots,\mathbf{P}[2],\mathbf{P}[0]]\\
    \mathbf{P_1} &= [\mathbf{P}[(|\mathbf{P}|/2 -1],\ldots,\mathbf{P}[3],\mathbf{P}[1]]\\
    \mathbf{P_2} &= [\mathbf{P}[(|\mathbf{P}| -2],\ldots,\mathbf{P}[|\mathbf{P}|/2 +2],\mathbf{P}[|\mathbf{P}|/2]]\\
    \mathbf{P_3} &= [\mathbf{P}[(|\mathbf{P}|-1],\ldots,|\mathbf{P}|/2+3],\mathbf{P}[|\mathbf{P}|/2+1]]
\end{aligned}
$$
The sub-sequence $\mathbf{P_0}$ (resp., $\mathbf{P_1}$) is assigned the even-index (resp., odd-index) processors in the first half of $\mathbf{P}$ (i.e., $\mathbf{P}'$), while $\mathbf{P_2}$ (resp., $\mathbf{P_3}$) is assigned the even-index (resp., odd-index) processors in the second half of $\mathbf{P}$ (i.e., $\mathbf{P}''$). In a BF step, \coim{} assigns each subsequence of the available processors to one of the four sub-problems which are recursively invoked to compute the product $A\times B$: $\mathbf{P_0}$ computes $A_0\times B_0$, $\mathbf{P_1}$ computes $A_0\times B_1$, $\mathbf{P_2}$ computes $A_1\times B_0$ and $\mathbf{P_3}$ computes $A_1\times B_1$.
\coim{} transfers to each sequence $\mathbf{P_i}$ the input integers for the corresponding sub-problem. This is achieved in the two following parallel communication steps: 
\begin{enumerate}
    \item In parallel, each odd-index (resp., even-index) processor $\mathbf{P}[i]$ (resp., $\mathbf{P}[j]$) for $i=1,3,\ldots,|\mathbf{P}|/2-1$ (resp., $j=|\mathbf{P}|/2,|\mathbf{P}|/2+2,\ldots,|\mathbf{P}|-2$) sends to $\mathbf{P}[i-1]$ (resp., $\mathbf{P}[j+1]$) a copy of the $n/|\mathbf{P}|$ digits of $A(\mathbf{P}[i])$ and $B(\mathbf{P}[i])$ (resp., $A(\mathbf{P}[j])$ and $B(\mathbf{P}[j])$). 
    At the end of this step $A_0$ and $B_0$ (resp., $A_1$ and  $B_11$) are partitioned in $\mathbf{P_0}$ (resp., $\mathbf{P_2}$) in $2n/|\mathbf{P}|$ digits. 
    \item In parallel, all processor $\mathbf{P_0}[i]$ (resp., $\mathbf{P_3}[i]$), for $i\in\{0,1,\ldots, |\mathbf{P}|/4\}$,  send to $\mathbf{P_1}[i]$ (resp., $\mathbf{P_2}[i]$) a copy of the digits of $A_0(\mathbf{P_0}[i])$ (resp.,  $A_1(\mathbf{P_2}[i])$).
    Ad the end of this step  a copy of $A_0$ and (resp., $A_1$) is partitioned in $\mathbf{P_1}$ (resp., $\mathbf{P_3}$) in $2n/|\mathbf{P}|$ digits.
    \item In parallel, all processor $\mathbf{P_0}[i]$ (resp., $\mathbf{P_3}[i]$), for $i\in\{0,1,\ldots, |\mathbf{P}|/4\}$, send  to $\mathbf{P_2}[i]$ (resp., $\mathbf{P_1}[i]$) a copy of the digits of $B_0(\mathbf{P_0}[i])$ (resp.,  $B_1(\mathbf{P_2}[i])$).
    Ad the end of this step  a copy of $B_0$ and (resp., $B_1$) is partitioned in $\mathbf{P_3}$ (resp., $\mathbf{P_1}$) in $2n/|\mathbf{P}|$ digits.
\end{enumerate}
 
 \item \textbf{Recursive multiplication:} \sloppy The four sub-products are computed in parallel by recursively invoking $\coim{}_{MI}$: $C_0= \textsc{COIM}_{MI}(\mathbf{P_0},A_0, B_0,n/|\mathbf{P}|)$,  $C_1= \textsc{COIM}_{MI}(\mathbf{P_1},A_0, B_1,n/|\mathbf{P}|)$ $\mathbf{P_0}$, $C_2= \textsc{COIM}_{MI}(\mathbf{P_2},A_1, B_0,n/|\mathbf{P}|)$ and $C_3= \textsc{COIM}_{MI}(\mathbf{P_3},A_1, B_1,n/|\mathbf{P}|)$. At the end of these recursive calls, $C_i$, for $i=0,1,2,3$, is partitioned in $\mathbf{P_i}$ in $4n/|\mathbf{P}|$ digits.
 
 \item \textbf{Recomposition:} $\coim{}_{MI}$ composes the values $C_0$, $C_1$,$C_2$ and $C_3$ to obtain $C$. First, the outputs of the sub-problems are opportunely redistributed among the processors: 
 \begin{itemize}
    \item[(a)] In parallel, each processor $\mathbf{P}_0[j]$ (resp., $\mathbf{P}_3[j]$) for $j=1,3,\ldots,|\mathbf{P}_0|= |\mathbf{P}_3| = |\mathbf{P}|/4$, sends to $\mathbf{P}_1[j]$ (resp., $\mathbf{P}_2[j]$) the $2n/|\mathbf{P}|$ most (resp., least) significant digits of $C_0(\mathbf{P}_1[j])$ (resp., $C_3(\mathbf{P}_3[j])$), and then it deletes them from its local memory. 
    At the end of this step $C_0$ (resp., $C_3$) is partitioned in $\mathbf{P}[|\mathbf{P}|/2-1..0]$ (resp., $\mathbf{P}[|\mathbf{P}|-1..|\mathbf{P}|/2]$) in $2n/|\mathbf{P}|$ digits.  
    
    \item[(b)] In parallel, each processor $\mathbf{P}_1[j]$ (resp., $\mathbf{P}_2[j]$) for $j=1,3,\ldots,|\mathbf{P}_1|= |\mathbf{P}_2| = |\mathbf{P}|/4$, sends to $\mathbf{P}_0[j]$ (resp., $\mathbf{P}_3[j]$) the $2n/|\mathbf{P}|$ least (resp., most) significant digits of $C_1(\mathbf{P}_1[j])$ (resp., $C_2(\mathbf{P}_3[j])$), and then it deletes them from its local memory. 
    At the end of this step $C_1$ (resp., $C_2$) is partitioned in $\mathbf{P}[|\mathbf{P}|/2-1..0]$ (resp., $\mathbf{P}[|\mathbf{P}|-1..|\mathbf{P}|/2]$) in $2n/|\mathbf{P}|$ digits.  
    
    \item[(c)] In parallel, each processor $\mathbf{P}[j+|\mathbf{P}|/4]$ for $j=0,1,\ldots,|\mathbf{P}|/4-1$, sends to $\mathbf{P}[j+|\mathbf{P}|/2]$ the $2n/|\mathbf{P}|$ digits of $C_1(\mathbf{P}[j+|\mathbf{P}|/4])$, and then it deletes them from its local memory.
    At the end of this step the $n/2$ most significant digits of $C_1$ are partitioned in $\mathbf{P}[3|\mathbf{P}|/4-1..|\mathbf{P}|/2]$ in $2n/|\mathbf{P}|$ digits.  
    
    \item[(d)] In parallel, each processor $\mathbf{P}[j+|\mathbf{P}|/2]$ for $j=0,1,\ldots,|\mathbf{P}|/4-1$, sends to $\mathbf{P}[j+|\mathbf{P}|/4]$ the $2n/|\mathbf{P}|$ digits of $C_2(\mathbf{P}[j+|\mathbf{P}|/4])$, and then it deletes them from its local memory.
    At the end of this step the $n/2$ least significant digits of $C_1$ are partitioned in $\mathbf{P}[|\mathbf{P}|/2-1..|\mathbf{P}|/4]$ in $2n/|\mathbf{P}|$ digits.  
    
    \item[(e)] In parallel, each processor $\mathbf{P}[j]$ (resp., $\mathbf{P}[j+3|\mathbf{P}|/4]$) for $j=0,1,\ldots,|\mathbf{P}|/4-1$, sends to $\mathbf{P}[j+|\mathbf{P}|/4]$ (resp., $\mathbf{P}[j+|\mathbf{P}|/2]$) the $2n/|\mathbf{P}|$ digits of $C_1(\mathbf{P}[j])$ (resp., $C_2(\mathbf{P}[j+3|\mathbf{P}|/4])$), and then it deletes them from its local memory.
 \end{itemize}
 After these five parallel communication steps, $C_0$ is partitioned in $[\mathbf{P}[|\mathbf{P}|/2-1],\ldots,\mathbf{P}[0]]$ in $2n/|\mathbf{P}|$ digits, $C_1$ and $C_2$ are partitioned in $[\mathbf{P}[3|\mathbf{P}|/4-1],\ldots,\mathbf{P}[|\mathbf{P}|/4]]$ in $2n/|\mathbf{P}|$ digits, and $C_3$ is partitioned  in $[\mathbf{P}[|\mathbf{P}|-1],\ldots,\mathbf{P}[|\mathbf{P}|/2]]$ in $2n/|\mathbf{P}|$ digits. We have $C(\mathbf{P}[j])= C_0((\mathbf{P}[j])$ for $j=0,1,\ldots, |\mathbf{P}|/4-1$. Let $C'_0$ denote the integer whose base-$s$ expansion corresponds to the $n/2$ most significant digits of $C_0$, and let $C'_3 = C_3*s^{n/2}$ (i.e., the integer whose base-$s$ expansion correspond to the digits of $C_3$ shifted $n/4$ times). By construction, $C'_0$ $C'_3$ (as well as $C_1$ and $C_2$) are partitioned in the subsequence $\mathbf{P}^*= [\mathbf{P}[|\mathbf{P}|-1],\ldots,3|\mathbf{P}|/4+1,3|\mathbf{P}|/4]$ in $2n/|\mathbf{P}|$ digits. The $3n/2$ most significant digit of $C$ correspond to those of the sum of $C'_0$, $C_1$,$C_2$ and $C'_3$, which can be computed with three consecutive invocations of the $\textsc{SUM}$ subroutine discussed in Section~\ref{sec:distsum} using the sequence $\mathbf{P}^*$. The $3n/2$ digits of such sum are partitioned in $P^*$ in $2n/|\mathbf{P}|$ digits, and, hence, $C$ is partitioned in $\mathbf{P}$ in $2n/|\mathbf{P}|$ digits.
\end{enumerate}

The following theorem characterizes memory utilization, computation cost, and \io{} cost for $\coim{}_{MI}$:

\begin{theorem}\label{thm:costcoimUM}
Let $A$ and $B$ be two $n$-digit integers portioned in a sequence of processors $\mathbf{P}$, with $n\geq |\mathbf{P}|$, in $n/|\mathbf{P}|$ digits. $\coim{}_{MI}$ computes the product $C=A\times B$ using the processors in $\mathbf{P}$, provided that each processor is equipped with a local memory of size at least $M_{\coim{}_{MI}}(n,|\mathbf{P}|)=12n/\sqrt{|\mathbf{P}|}$.  We have:
\begin{align*}
    T_{\coim{}_{MI}}\left(n,|\mathbf{P}|\right) &\leq 38\frac{n^2}{|\mathbf{P}|} + 3\log_2^2 |\mathbf{P}|\\
    BW_{\coim{}_{MI}}\left(n,|\mathbf{P}|\right) &\leq 14\frac{n}{\sqrt{|\mathbf{P}|}} + 6\log^2_2 |\mathbf{P}|\\
    L_{\coim{}_{MI}}\left(n,|\mathbf{P}|\right) &\leq 3 \log^2_2 |\mathbf{P}|
\end{align*}
\end{theorem}
\begin{proof}
$\coim{}_{MI}$ correctly computes $C=A\times B$ by inspection.

If $|\mathbf{P}|=1$, the product $C$ is computed locally by the single available processor using algorithm $\textsc{SLIM}$. By Lemma~\ref{lem:seqcostlim} we thus have $M_{\coim{}_{MI}}(n,1)\le 8n$,  $T_{\coim{}_{MI}}(n,1)\le 8n^2$, $BW_{\coim{}_{MI}}(n,1)=0$ and $L_{\coim{}_{MI}}(n,1)=0$.

In the following, we assume $|\mathbf{P}|>1$. By construction, during the $i$-th recursion step of $\coim{}_{MI}$ for $1\leq i < \log_4 |\mathbf{P}|$, each processor needs to hold in its local memory at most $2\times 4^{i-1}n/|\mathbf{P}|2^{i-1} = n2^{i}/|\mathbf{P}|$ digits of the inputs one of the sub-problems being computed at the $i$-th level, and $n2^{i+1}/|\mathbf{P}|$ digits of the input of the sub-problem being generated to whom the processor will be assigned. Hence, $M_{\coim{}_{MI}}(n,|\mathbf{P}|)\geq 3n/\sqrt{|\mathbf{P}|}$. Similarly, once the output of the sub-problems generated at the $i$-th recursion step of $\coim{}_{MI}$, for $1< i \le \log_4 |\mathbf{P}|$, have been computed, each processor needs to hold in its local memory at most $3\times n2^{i}/|\mathbf{P}|$ digits of the outputs of one of the sub-problems computed at the $i$-th level, and $n2^{i-1}/|\mathbf{P}|$ digits of the output of one the sub-problem at the $(i-1)$-th level. Hence, $M_{\coim{}_{MI}}(n,|\mathbf{P}|)\geq 4n/\sqrt{|\mathbf{P}|}$. Further, each processor must be equipped with enough memory to sum $C'_0$, $C_1$, $C_2$ and $C'_3$ using \textsc{SUM}, as discussed at end of phase 3 of $\coim{}_{MI}$'s description. In the $i$-th recursion level, for $1\le i \le \log_4 |\mathbf{P}|-1$, \textsc{SUM} is used to sum $(3n2^{i+1}/|\mathbf{P}|)$-digit integers initially partitioned in $3|\mathbf{P}|/4^{i+1}$ processors (which are used to compute the sum) in $2^{i}n/|\mathbf{P}|$ digits. Hence, from Lemma~\ref{lem:costsum}, $M_{\coim{}_{MI}}(n,|\mathbf{P}|)\geq 4(2^{i+1}n/|\mathbf{P}|+1)\leq 12n/\sqrt{|\mathbf{P}|}$.
Finally, each processor must be equipped with enough memory to compute the product of two $(n/\sqrt{|\mathbf{P}|})$ digits integers using algorithm $\textsc{SLIM}$. Hence, by Lemma~\ref{lem:seqcostlim}, $M_{\coim{}_{MI}}(n,|\mathbf{P}|)\geq 8n/\sqrt{|\mathbf{P}|}$. We can thus conclude that all such requirement are met if each processor is equipped with a local memory of size $M_{\coim{}_{MI}}(n,|\mathbf{P}|)=12n/\sqrt{|\mathbf{P}|}$.\\

In the ``\emph{Recomposition}'' step $\coim{}_{MI}$ computes the sum of $C'_0$,$C_1$,$C_2$, $C'_3$ using  the $\textsc{SUM}$ parallel subroutine. By construction, each such values has at most $3n/2$ digits and it is partitioned in $\mathbf{P}^*$ in $2n/|\mathbf{P}|$ digits. Their sum can be computed with three consecutive applications of $\textsc{SUM}$ using the subsequence $\mathbf{P}^*$. As no other computation is executed outside the recursive calls to \coim{},  by Lemma~\ref{lem:costsum} and as $|\mathbf{P}^*| = 3\mathbf{P}|/4$, we have:
\begin{equation*}\label{eq:time1coimum}
\begin{aligned}
    T_{\coim{}_{MI}}\left(n,|\mathbf{P}|\right) &\leq  T_{\coim{}_{MI}}\left(n/2, |\mathbf{P}|/4\right) + 3T_{\textsc{SUM}}\left(3n/2,|\mathbf{P}^* | \right)\\
    &\leq T_{\coim{}_{MI}}\left(n/2, |\mathbf{P}|/4\right) + 3\left((12n/|\mathbf{P}| + 4\log_2 (3|\mathbf{P}|/4-1)\right)\\
\end{aligned}
\end{equation*}
In phase (1a), each of the $\npr{P}$ processor either sends or receives $2n/\npr{P}$. In each of the phases (1b) and (1c), each processor which communicates either sends or receives $n/\npr{P}$ digits. During each of the phases from (3a) to (3e) every processor which communicates either sends or it receives $2n/\npr{P}$ memory words to/from other processors. As the four generated subproblems, each with input size at most $n/2$, are computed in parallel using $\npr{P}{}/4$ processors each, by Lemma~\ref{lem:costsum} and as $|\mathbf{P}^*| = 3|\mathbf{P}|/4$, we have
\begin{equation*}\label{eq:band1coimum}
    \begin{aligned}
        BW_{\textsc{COIM}_{MI}}\left(n,|\mathbf{P}|\right) &\leq  BW_{\textsc{COIM}_{MI}}\left(n/2,|\mathbf{P}|/4\right) + 14n/|\mathbf{P}| + 3BW_{\textsc{SUM}}\left(3n/2, |\mathbf{P}^*|\right)\\
        &\leq  BW_{\textsc{COIM}_{MI}}\left(n/2,|\mathbf{P}|/4\right) + 14n/|\mathbf{P}| + 12(\log_2 3|\mathbf{P}|/4-1)\\
        L_{\textsc{COIM}_{MI}}\left(n,|\mathbf{P}|\right) &\leq  L_{\textsc{COIM}_{MI}}\left(n/2,|\mathbf{P}|/4\right) + 8 + 3L_{\textsc{SUM}}\left(3n/2, |\mathbf{P}^*|\right)\\
        &\leq  L_{\textsc{COIM}_{MI}}\left(n/2,|\mathbf{P}|/4\right) + 8 + 6(\log_2 3|\mathbf{P}|/4-1)\\
    \end{aligned}
\end{equation*}
After $\log_4 |\mathbf{P}|$ recursive breadth-first levels, each processor is assigned a single subproblem with input size $n/2^{\log_4 |\mathbf{P}|} = n/\sqrt{|\mathbf{P}|}$, which is computed \emph{locally} without any further communication as discussed for the base case (i.e., $BW_{\textsc{COIM}_{MI}}\left(n/|\mathbf{P}|,1\right)=L_{\textsc{COIM}_{MI}}\left(n/|\mathbf{P}|,1\right)=0$) using a \emph{sequential} long integer multiplication algorithm (e.g., algorithm $\textsc{SLIM}$). Hence, by Lemma~\ref{lem:seqcostlim}, $T_{\textsc{COIM}_{MI}}\left(n/|\mathbf{P}|,1\right) \leq 2n^2/|\mathbf{P}|$. As, by assumption, $n\geq |\mathbf{P}|$ we have:
\begin{align*}
    T_{\coim{}_{MI}}\left(n,|\mathbf{P}|\right) &\leq  T_{\coim{}_{MI}}\left(\frac{n}{\sqrt{|\mathbf{P}|}}, 1\right) + 36\frac{n}{|\mathbf{P}|}\sum_{i=1}^{\log_4 |\mathbf{P}|-1}2^i +   6\sum_{i=1}^{\log_4 |\mathbf{P}|-1} \log_2 \frac{3|\mathbf{P}|}{4^i}-1\\
    &< 2\frac{n^2}{|\mathbf{P}|} + 36\frac{n}{\sqrt{|\mathbf{P}|}} + 3\log_2^2 |\mathbf{P}|\\
    &= 38\frac{n^2}{|\mathbf{P}|} + 3\log_2^2 |\mathbf{P}|\\
        BW_{\textsc{COIM}_{MI}}\left(n,|\mathbf{P}|\right) &\leq BW_{\textsc{COIM}_{MI}}\left(\frac{n}{\sqrt{|\mathbf{P}|}},1\right)+ 14\frac{n}{|\mathbf{P}|}\sum_{i=1}^{\log_4 |\mathbf{P}|-1}2^i + 12\sum_{i=1}^{\log_4 |\mathbf{P}|-1}(\log_2 \frac{3|\mathbf{P}|}{4^i}-1)\\
        &< 14\frac{n}{\sqrt{|\mathbf{P}|}} + 6\log^2_2 |\mathbf{P}|\\
        L_{\textsc{COIM}_{MI}}\left(n,|\mathbf{P}|\right) &\leq  L_{\textsc{COIM}_{MI}}\left(n/\sqrt{|\mathbf{P}|},1\right) + \sum_{i=1}^{\log_4 |\mathbf{P}|-1} 8+ 6\sum_{i=1}^{\log_4 |\mathbf{P}|-1}(\log_2 3|\mathbf{P}|/4^i-1)\\
        &< 3 \log^2_2 |\mathbf{P}|
\end{align*}
\end{proof}
\subsection{\coim{} in the main execution mode}\label{sec:coimLM}
An important consequence of Theorem~\ref{thm:costcoimUM} is that to execute $\coim{}_{MI}$ to multiply two $n$-digit integers each processor must be equipped with a memory of size at least $12n/\sqrt{|\mathbf{P}|}$. Considering the aggregate memory available in $|\mathbf{P}|$ processors, this implies that $\coim{}_{MI}$ can only be used to multiply $n$-digit integers where $n\in \BO{M\sqrt{|\mathbf{P}|}}$.  In this section, we describe how to combine the MI execution mode  $\coim{}_{MI}$ with a \emph{depth-first} scheduling of the subproblems, as outlined in the general framework in Section~\ref{sec:overview}. Our algorithm \coim{} allows to compute the product of two $n$-digit integers provided that (i) $M\geq 80n/|\mathbf{P}|$ and (ii) $M\geq \log_2 \npr{P}$. By (i), \coim{} can thus be used to multiply $n$-digit integers where $n\in \BO{M|\mathbf{P}|}$. That is, its memory requirement corresponds asymptotically to the amount of overall memory required to store the input integers and/or the product. Requirement $(ii)$ is also very reasonable in practice, as it is generally more sensible and cost/effective to have each processor equipped with a memory of respectable size rather than explode the number of available processors. 

\sloppy In the main execution mode, \coim{} proceeds in a sequence of at most $\BO{\log_2 n/(M\sqrt{|\mathbf{P}|})}$ recursive \emph{depth-first steps}, where $M$ denotes the size of the memory available to each processor, until the size of the generated sub-problems allows them to be computed in the MI execution mode $\coim{}_{MI}$.\\


In the main execution mode, when $\coim{}\left(n,A,B,\mathbf{P},M, n/|\mathbf{P}|\right)$ is invoked, if $M\geq  12n\sqrt{|\mathbf{P}|}$, that is if the size of the memory available to each processor allows for it, the product is computed by invoking $\coim{}_{MI}{}\left(n,A,B,\mathbf{P},n/|\mathbf{P}|\right)$. Otherwise \coim{} will proceed with a \emph{depth-first step}, by generating four subproblems invoking itself on each of the subproblems one at a time. The memory available to the recursive calls is reduced by the amount of memory space required to maintain the data used in the depth first step, including, for each processor, $2n/|\mathbf{P}|$ digits of the input, up to $3n/|\mathbf{P}|$ digits of the outputs of the sub-problems, and the memory space required for an invocation of $\textsc{SUM}$ to sum $(3n/2)$-digits integers partitioned in $3|\mathbf{P}|/4$ processors (i.e., by Lemma~\ref{lem:costsum}, $4(2n/|\mathbf{P}|+1)$). Hence, the available memory is reduced by $20n/|\mathbf{P}|$.  In the following, we denote $A_0$,$B_0$, $A_1$ and $B_1$ as defined in~\eqref{eq:partitionAB2}. 
The operations being executed are slightly different depending on the subproblems:
\begin{itemize}
    \item $\mathbf{\mathbf{C_0 = A_0\times B_0}:}$ In parallel, each processor $\mathbf{P}[j]$ for $j=0,1,\ldots,|\mathbf{P}|/2-1$ sends to $\mathbf{P}[j+|\mathbf{P}|/2]$ a copy of the $n/(2|\mathbf{P}|)$ most significant digits of $A(\mathbf{P}[j])$ and $B(\mathbf{P}[j])$. At the end of this step $A_0$ and $B_0$  are partitioned in the sequence
    $$\mathbf{P'}= [\mathbf{P}[|\mathbf{P}|-1], \mathbf{P}[|\mathbf{P}|/2-1], \ldots,\mathbf{P}[|\mathbf{P}|/2+1],\mathbf{P}[1]], \mathbf{P}[|\mathbf{P}|/2],\mathbf{P}[0]]$$
    in $n/(2|\mathbf{P}|)$ digits (i.e., the even/odd index processors in the sequence $\mathbf{P'}$ are the processors in the first/second half of the sequence $\mathbf{P}$).
    $\coim{}(n,A_0,B_0,\mathbf{P},M-20n/|\mathbf{P}|,n/(2|\mathbf{P}|))$ is then recursively invoked to compute $C_0=A_0\times B_0$ using the sequence $\mathbf{P}'$. Once computed, $C_0$ is partitioned in $\mathbf{P'}$ in $n/|\mathbf{P}|$ digits. 
    
    Before proceeding in the computation of the second subproblem, \coim{} rearranges the digits of the output $C_0$ so that it is partitioned in the fist half of the processors sequence $\mathbf{P}$, that is  $[\mathbf{P}[|\mathbf{P}|/2-1],\ldots,\mathbf{P}[0]]$ in $n/(2|\mathbf{P}|)$ digits. This is accomplished in a single communication step during which in parallel, each odd-index processor of $\mathbf{P}'[j]$, for $j=1,3,\ldots,|\mathbf{P}'|/2-1$, sends to $\mathbf{P}'[j-1]$ the $n/|\mathbf{P}|$ digits of $C_0(\mathbf{P}'[j])$, and then it deletes them from its local memory.
    
    \item $\mathbf{C_1 = A_0\times B_1:}$ In parallel, each processor $\mathbf{P}[j]$ for $j=0,1,\ldots,|\mathbf{P}|/2-1$ sends to $\mathbf{P}[j+|\mathbf{P}|/2]$ a copy of the $n/(2|\mathbf{P}|)$ most significant digits of $A(\mathbf{P}[j])$, and then removes those digits from its local memory. Then, in parallel, each processor $\mathbf{P}[j+ |\mathbf{P}|/2]$ for $i=0,1,\ldots,|\mathbf{P}|/2-1$ sends to $\mathbf{P}[j]$ a copy of the $n/(2|\mathbf{P}|)$ least significant digits of $B(\mathbf{P}[j+ |\mathbf{P}|/2])$. At the end of these two steps, $A_0$ and $B_1$ are partitioned in the sequence $\mathbf{P'}$ in $n/(2|\mathbf{P}|)$ digits. $\coim{}(n,A_0,B_1,\mathbf{P},M-20n/|\mathbf{P}|,n/(2|\mathbf{P}|))$ is then recursively invoked to compute $C_1=A_0\times B_1$ using the sequence $\mathbf{P}'$. Once computed, $C_1$ is partitioned in $\mathbf{P'}$ in $n/|\mathbf{P}|$ digits. 
   
   Before proceeding in the computation of the third subproblem, \coim{} rearranges the digits of the output $C_1$ so that it is partitioned in the sub-sequence of $\mathbf{P}$ composed of the $|\mathbf{P}|/2$ processors in the ``\emph{middle}'' of $\mathbf{P}$, $[\mathbf{P}[3|\mathbf{P}|/4-1],\ldots,\mathbf{P}[|\mathbf{P}|/4]]$, in $2n/|\mathbf{P}|$ digits.
     In parallel, each odd-index (resp., even-index) processor $\mathbf{P}'[j]$ (resp., $\mathbf{P}'[j+|\mathbf{P}'|/2-1]$) for $j=1,3,\ldots,|\mathbf{P}'|/2-1$ sends to $\mathbf{P}[j-1]$ (resp., $\mathbf{P}[j+|\mathbf{P}'|/2]$) a copy of the $n/|\mathbf{P}|$ digits of $C_1(\mathbf{P}[j])$ (resp., $C_1(\mathbf{P}[j+|\mathbf{P}'|/2-1])$), and then deletes them from its local memory. After this step, the $n/2$ least (resp., most) significant digits of $C_1$ are partitioned in the sub-sequence  $[\mathbf{P}[|\mathbf{P}|/4-1],\ldots,\mathbf{P}[0]]$ (resp., $[\mathbf{P}[|\mathbf{P}|-1],\ldots,\mathbf{P}[3|\mathbf{P}|/4]]$) in $2n/|\mathbf{P}|$ digits. In a following parallel communication step, in parallel, each processor $\mathbf{P}[j]$ (resp., $\mathbf{P}[j+3|\mathbf{P}|/4]$), for $j=0,1,\ldots, |\mathbf{P}|/4-1$, send to $\mathbf{P}[j+|\mathbf{P}|/4]$ (resp., $\mathbf{P}[j+|\mathbf{P}|/2]$) the $2n/|\mathbf{P}|$ digits of $C_1(P[j])$ (resp., $C_1(P[j+3|\mathbf{P}|/4])$), and then removes them from its local memory. 
     \item $\mathbf{C_2 = A_1\times B_0:}$ The operations executed for this sub-problem closely follow those discussed in the previous one. A detailed description can be obtained by replacing $B_0$ with $A_0$, $A_1$ with $B_1$, and $C_2$ with $C_1$. 
     \item $\mathbf{C_3=A_1\times B_1:}$ In parallel, each processor $\mathbf{P}[j+|\mathbf{P}|/2]$ for $j=0,1,\ldots,|\mathbf{P}|/2-1$ sends to $\mathbf{P}[j]$ a copy of the $n/(2|\mathbf{P}|)$ least significant digits of $A(\mathbf{P}[i])$ and $B(\mathbf{P}[i])$. At the end of this step $A_1$ and $B_1$  are partitioned in the sequence $\mathbf{P'}$ in $n/(2|\mathbf{P}|)$ digits. \coim{} is then recursively invoked to compute $C_3=A_1\times B_1$ using the sequence $\mathbf{P}'$. Once computed, $C_3$ is partitioned in $\mathbf{P'}$ in $n/|\mathbf{P}|$ digits. 
     
     Before proceeding, \coim{} rearranges the digits of the output $C_3$ so that it is partitioned in the second half of the processors sequence $\mathbf{P}$, that is  $[\mathbf{P}[|\mathbf{P}|-1],\ldots,\mathbf{P}[|\mathbf{P}|/2]]$. This is accomplished in a single communication step during which in parallel, each even-index processor of $\mathbf{P}'[j]$, for $j=0,2,\ldots,|\mathbf{P}'|/2-2$, sends to $\mathbf{P}'[j+1]$ the $n/|\mathbf{P}|$ digits of $C_0(\mathbf{P}'[j])$, and then it deletes them from its local memory.
\end{itemize}
     
    After the four sub-problems have been computed, and their respective outputs rearranged in $\mathbf{P}$. \coim{} completes the computation of $C$ following the same steps presented for the MI execution mode memory setting at the end of phase (3). Once computed, $C$ is correctly partitioned in $\mathcal{C}$ in $2n/|\mathbf{P}|$ digits. 

\begin{theorem}\label{thm:coim}
Let $A$ and $B$ be two $n$-digit integers portioned in a sequence of processors $\mathbf{P}$, with $n\geq |\mathbf{P}|$, in $n/|\mathbf{P}|$ digits. $\coim{}_{MI}$ computes the product $C=A\times B$ using the processors in $\mathbf{P}$, provided that each processor is equipped with a local memory of size at least $M_{\coim{}_{MI}}(n,|\mathbf{P}|)\geq \max\{80n/|\mathbf{P}|,\log_2 \npr{P}\}$.  We have:
\begin{align*}
    T_{\coim{}}\left(n,|\mathbf{P}|,M\right) &\leq 196\frac{n^2}{|\mathbf{P}|}\\
    BW_{\coim{}}\left(n,|\mathbf{P}|,M\right) &\leq  3530\frac{n^2}{M|\mathbf{P}|}\\
    L_{\coim{}}\left(n,|\mathbf{P}|,M\right) &\leq 7012\frac{n^2 \log_2^2|\mathbf{P}|}{M^2|\mathbf{P}|}
\end{align*}
\end{theorem}
\begin{proof}
$\coim{}$ correctly computes $C=A\times B$ by inspection.

If $n\leq M\sqrt{|\mathbf{P}|}/12$, the statement follows from Theorem~\ref{thm:costcoimUM}. In the following we assume $M|\mathbf{P}|/80\geq n>M\sqrt{|\mathbf{P}|}/12$. 
\coim{} then proceeds by executing $\ell$ consecutive depth-first steps, where $0< \ell\leq \lceil \log_2 n/(24M\sqrt{P})\rceil$. In each recursion step, in addition to the space required for the recursive invocation to $\coim{}$, each processor must maintain $2\times n/(|\mathbf{P}|)$ digits for the input of the problem being computed, at most $6n/|\mathbf{P}|$ digits of the outputs of the recursive subproblems, and the space required for the invocation of $\textsc{SUM}$ used to combine the outputs of the subproblems. Hence:
\begin{align}
        M_{\coim}\left(n,|\mathbf{P}|\right) &\leq M_{\coim}\left(\frac{n}{2},|\mathbf{P}|\right)+ 2\frac{n}{|\mathbf{P}|} +6\frac{n}{|\mathbf{P}|} + M_{SUM} (3n/2, 3|\mathbf{P}|/4)\nonumber\\
        &\leq M_{\coim}\left(\frac{n}{2},|\mathbf{P}|\right)+ 8\frac{n}{|\mathbf{P}|}+4\left(2\frac{n}{|\mathbf{P}|}+1\right)\label{eq:coimb1}\\
        &\leq M_{\coim}\left(\frac{n}{2},|\mathbf{P}|\right)+ 8\frac{n}{|\mathbf{P}|}+12\frac{n}{|\mathbf{P}|}\nonumber\\
        &\leq M_{\coim}\left(\frac{M\sqrt{|\mathbf{P}|}}{24},|\mathbf{P}|\right) +20\frac{n}{|\mathbf{P}|}\sum_{\ell=1}^{\lceil \log_2 \frac{24n}{M\sqrt{|\mathbf{P}|}}\rceil-1} 2^{-1}\nonumber\\
        &< M_{\coim{}_{MI}}\left(\frac{M\sqrt{|\mathbf{P}|}}{24},|\mathbf{P}|\right) +40\frac{n}{|\mathbf{P}|}\nonumber\\
       &\leq \frac{M}{2}+40\frac{n}{|\mathbf{P}|}\label{eq:coimb3}
\end{align}
where~\eqref{eq:coimb1} follows from Lemma~\ref{lem:costsum},
and~\eqref{eq:coimb3} follows from Theorem~\ref{thm:costcoimUM}. 
By construction, after $\ell$ depth-first steps, the available memory is reduced by at most $40n/|\mathbf{P}|$. As, by assumption, $M\geq 80 n/|\mathbf{P}|$, at least half of the space $M$ originally assigned to each of the processors is still available. After at most $\lceil \log_2 24n/(M\sqrt{P})\rceil$ recursive steps, the generated sub-problems will have size at most $M\sqrt{|\mathbf{P}|}/24$. Thus, by Theorem~\ref{thm:costcoimUM}, they can be computed using $\coim{}_{MI}$ using at most $M/2$ memory locations. This concludes the proof of the memory requirement for \coim{}.

The computation time required in a depth-first recursion level of \coim{} execution is bounded by the time required for the sequential computation of \coim{}'s recursive invocation on the generated subproblems plus the time required to combine the outputs of the sub-problems to compute the product itself using three invocations of $\textsc{SUM}$. As the numbers being summed have at most $3n/(2|\mathbf{P}|)$ each and are partitioned in the $3|\mathbf{P}|/4$ processors used to sum them in $2n/|\mathbf{P}|$ digits, by Lemma~\ref{lem:costsum}, we have:
\begin{align*}
    T_{\coim{}}\left(n, |\mathbf{P}|,M\right) &\leq 4 T_{\coim{}}\left(\frac{n}{2}, |\mathbf{P}|,M\right) + 3 T_{\textsc{SUM}}\left(\frac{3n}{2},\frac{3|\mathbf{P}|}{4}\right)\\
    &\leq 4 T_{\coim{}}\left(\frac{n}{2}, |\mathbf{P}|,M \right) + 3\left( \frac{2n}{|\mathbf{P}|} + 4\log_2 |\mathbf{P}|\right)
\end{align*}
Further, in a depth-first recursion level the \io{} cost of \coim{} (both bandwidth and latency) can be bound by that of the four consecutive invocations of $\coim{}$ used to compute the four subproblems, the cost of redistributing the input (resp., the output) of such subproblems, and the \io{} cost of the three invocations of $\textsc{SUM}$ used to combine the outputs of the three subproblems. By Lemma~\ref{lem:costsum}, we have:
\begin{equation*}
    \begin{aligned}
        BW_{\coim{}}\left(n, |\mathbf{P}|, M \right) 
        &\leq 4BW_{\coim{}}\left(\frac{n}{2}, |\mathbf{P}|,M \right) + 3\frac{n}{|\mathbf{P}|}+8\frac{n}{|\mathbf{P}|}+ 3 BW_{\textsc{SUM}}\left(\frac{3n}{2},3|\mathbf{P}|\right)\\
        &\leq 4BW_{\coim{}}\left(\frac{n}{2}, |\mathbf{P}|,M \right)+11\frac{n}{|\mathbf{P}|}+12\log_2|\mathbf{P}|\\\\
        L_{\coim{}}\left(n, |\mathbf{P}|, M \right) &\leq 4L_{\coim{}}\left(\frac{n}{2}, |\mathbf{P}|,M \right) + 12+ 3 L_{\textsc{SUM}}\left(\frac{3n}{2},\frac{3|\mathbf{P}|}{4}\right)\\
        &\leq 4L_{\coim{}}\left(\frac{n}{2}, |\mathbf{P}|,M \right)+12+6\log_2|\mathbf{P}|
    \end{aligned}
\end{equation*}
After $1\le\ell \leq \log_2 24n/(M\sqrt{|\mathbf{P}|})$ depth-first steps, the generated sub-problems have input size at most $M\sqrt{|\mathbf{P}|}/24$, and \coim{} switches to the MI execution mode by invoking $\coim{}_{MI}$. Hence, by Theorem~\ref{thm:coim}, and by the assumptions $n\geq |\mathbf{P}|$ and $M\geq \log_2 |\mathbf{P}|$  we have:
\begin{align*}
    T_{\coim{}}\left(n, |\mathbf{P}|,M\right) 
    &\leq 4^{\lceil \log_2 \frac{24n}{M\sqrt{|\mathbf{P}|}}\rceil}T_{\coim{}}\left(\frac{M\sqrt{P}}{24},|\mathbf{P}|,M\right) + \frac{6n}{|\mathbf{P}|}\sum_{\ell=1}^{\lceil \log_2 \frac{24n}{M\sqrt{|\mathbf{P}|}}\rceil-1}2^{-\ell}\\&\quad + \left(\lceil \log_2 \frac{24n}{M\sqrt{|\mathbf{P}|}}\rceil-1\right)4\log_2 |\mathbf{P}|\nonumber\\
    &< 4\times\frac{24^2n^2}{M^2|\mathbf{P}|}T_{\coim{}_{MI}}\left(\frac{M\sqrt{P}}{24},|\mathbf{P}|\right) + \frac{12n}{|\mathbf{P}|}+ \left(\lceil \log_2 \frac{24n}{M\sqrt{|\mathbf{P}|}}\rceil-1\right)4\log_2 |\mathbf{P}|\nonumber\\
    &\leq 4\times\frac{24^2n^2}{M^2|\mathbf{P}|}\left(\frac{38M^2}{24^2}+3\log_2^2 |\mathbf{P}|\right) + \frac{12n}{|\mathbf{P}|}+ \left(\lceil \log_2 \frac{24n}{M\sqrt{|\mathbf{P}|}}\rceil-1\right)4\log_2 |\mathbf{P}|\nonumber\\
    &\leq164\frac{n^2}{|\mathbf{P}|} + \frac{12n}{|\mathbf{P}|}+ \left(\lceil \log_2 \frac{24n}{M\sqrt{|\mathbf{P}|}}\rceil-1\right)4\log_2 |\mathbf{P}|\nonumber\\
    &< 196\frac{n^2}{|\mathbf{P}|}\nonumber
\end{align*}
\begin{align}
    BW_{\coim{}}\left(n, |\mathbf{P}|,M\right) 
    &\leq 4^{\lceil \log_2 \frac{24n}{M\sqrt{|\mathbf{P}|}}\rceil}BW_{\coim{}}\left(\frac{M\sqrt{P}}{24},|\mathbf{P}|,M\right) + 11\frac{n}{|\mathbf{P}|}\sum_{\ell=0}^{\lceil \log_2 \frac{24n}{M\sqrt{|\mathbf{P}|}}\rceil-1}2^{-\ell}\nonumber \\&\quad + 12\log_2 |\mathbf{P}|\left(\lceil \log_2 \frac{24n}{M\sqrt{|\mathbf{P}|}}\rceil-1\right)\nonumber\\
    &< 4\times\frac{24^2n^2}{M^2|\mathbf{P}|}BW_{\coim{}_{MI}}\left(\frac{M\sqrt{P}}{24},|\mathbf{P}|\right) + 22\frac{n}{|\mathbf{P}|}\nonumber\\
    &\quad+ 12\log_2 |\mathbf{P}|\left(\lceil \log_2 \frac{24n}{M\sqrt{|\mathbf{P}|}}\rceil-1\right)\nonumber\\
    &\leq 4\times\frac{24^2n^2}{M^2|\mathbf{P}|}\left(\frac{14M}{24^2}+6\log_2^2 |\mathbf{P}|\right) + \frac{22n}{|\mathbf{P}|}\nonumber\\&\quad+ 12\log_2 |\mathbf{P}|\left(\lceil \log_2 \frac{24n}{M\sqrt{|\mathbf{P}|}}\rceil-1\right)\nonumber\\
    &\leq 4\times\frac{24^2n^2}{M^2|\mathbf{P}|}M\frac{14+6\times24^2}{24^2} + \frac{22n}{|\mathbf{P}|}+ 12\log_2 |\mathbf{P}|\left(\lceil \log_2 \frac{24n}{M\sqrt{|\mathbf{P}|}}\rceil-1\right)\nonumber\\
    &\leq 3470\frac{n^2}{M|\mathbf{P}|} + \frac{22n}{|\mathbf{P}|}+ 12\log_2 |\mathbf{P}|\left(\lceil \log_2 \frac{24n}{M\sqrt{|\mathbf{P}|}}\rceil-1\right)\label{eq:ad1}\\
    &< 3530\frac{n^2}{M|\mathbf{P}|}\nonumber
\end{align}
\begin{align}
    L_{\coim{}}\left(n, |\mathbf{P}|,M\right) 
    &\leq 4^{\lceil \log_2 \frac{24n}{M\sqrt{|\mathbf{P}|}}\rceil}L_{\coim{}}\left(\frac{M\sqrt{P}}{24},|\mathbf{P}|,M\right)\nonumber\\&\quad+ \left(\lceil \log_2 \frac{24n}{M\sqrt{|\mathbf{P}|}}\rceil-1\right)\left( 12 +6\log_2 |\mathbf{P}|\right)\nonumber\\
    &< 4\times\frac{24^2n^2}{M^2|\mathbf{P}|}L_{\coim{}_{MI}}\left(\frac{M\sqrt{P}}{24},|\mathbf{P}|\right)\nonumber\\&\quad+ \left(\lceil \log_2 \frac{24n}{M\sqrt{|\mathbf{P}|}}\rceil-1\right)\left( 12 +6\log_2 |\mathbf{P}|\right)\nonumber\\
    &\leq 4\times\frac{24^2n^2}{M^2|\mathbf{P}|}3\log_2^2|\mathbf{P}| + \left(\lceil \log_2 \frac{24n}{M\sqrt{|\mathbf{P}|}}\rceil-1\right)\left( 12 +6\log_2 |\mathbf{P}|\right)\nonumber\\
    &\leq 6912\frac{n^2 \log_2^2|\mathbf{P}|}{M^2|\mathbf{P}|}+ \left(\lceil \log_2 \frac{24n}{M\sqrt{|\mathbf{P}|}}\rceil-1\right)\left( 12 +6\log_2 |\mathbf{P}|\right)\label{eq:ad2}\\
    &\leq 7012\frac{n^2 \log_2^2|\mathbf{P}|}{M^2|\mathbf{P}|} \nonumber
\end{align}
where~\eqref{eq:ad1} and~\eqref{eq:ad2} where the last passage follows as, by assumption, $n\geq \npr{P}$  and $M\geq \log_2 \npr{P}$, and as we are considering the case $n\geq M\sqrt{\npr{P}}/12$. If that was not the case the product would have been computed using $\coim_{MI}$.
\end{proof}

\subsection{Comparison with communication lower bounds}\label{sec:stalwbcomp}
Based on the analysis of \coim{} performance presented in Theorem~\ref{thm:costcoimUM} and Theorem~\ref{thm:coim}, we have:
\begin{reptheorem}{thm:informalcoim}
\coim{} achieves optimal computation time speedup and optimal bandwidth cost among all parallel standard integer multiplication algorithms. It also minimizes the latency cost up to a $\BO{\log^2\mathcal{P}}$ multiplicative factor. 
\end{reptheorem}
\begin{proof}
Let $\mathcal{P}$ denote the number of processors used in the computation.
By Theorem~\ref{thm:costcoimUM}, for $M\geq 12n/\sqrt{\mathcal{P}}$, the product $C=A\times B$ can be computed using $\coim{}_{MI}$. Under the assumptions $n\geq \mathcal{P}$ and $M\geq \log_2 \npr{P}$, the bandwidth cost of $\coim{}_{MI}$ asymptotically matches, the memory-independent lower bound in Theorem~\ref{thm:lwbstamemind}, and its latency latency is within a $\BO{\log^2 \mathcal{P}}$ factor of the corresponding lower bound. Note that the initial distribution of the input values among the processors used in $\coim{}_{MI}$ satisfies the balanced input distribution assumption used to derive  Theorem~\ref{thm:lwbstamemind}.

For $ 12n/\sqrt{\mathcal{P}}>M\geq 80n/\mathcal{P}$, by Theorem~\ref{thm:coim}, the product $C=A\times B$ can be computed using $\coim{}$. For $n\geq P$ and $M\geq 24\sqrt{\mathcal{P}}$, the bandwidth cost of $\coim{}_{MI}$ asymptotically matches, the memory-dependent lower bound in Theorem~\ref{thm:lwbstamemdip}, and its latency latency is within a $\BO{\log^2 P}$ factor of the corresponding lower bound. The total memory space required across the available processors for the execution of \coim{} is $\BO{n}$, that is, within a constant factor of the minimum space required to store the input (and output) values. Finally, in both cases, \coim{} achieves optimal computational time speedup $\BO{n^2/\mathcal{P}}$.
\end{proof}

\section{Communication-optimal Parallel Karatsuba \\ Multiplication}\label{sec:karatsuba}
In this section, we present $\cok{}$ (Communication Optimal Parallel Karatsuba), our parallel implementation of Karatsuba's algorithm in the distributed memory setting which achieves optimal speedup, optimal bandwidth cost, and whose latency is within $\BO{\log_2 |\mathbf{P}|}$ of the corresponding lower bound.
Compared to standard \emph{long} integer multiplication algorithms (e.g., \textsc{SLIM}), Karatsuba's algorithm asymptotically reduces the computation time of integer multiplication by cleverly decreasing the number of sub-problems being generated at each recursion level from four to three. 

While several different variations of the algorithm have been discusses in the literature, in this work we consider the following sequential implementation, henceforth referred to as $\textsc{SKIM}$ (Sequential Karatusba Integer Multiplication). The algorithm follows a simple recursive strategy. Let $A,B$ be $n$-digit integers, if $n=1$, the product $C=A\times B$ is computed directly. Otherwise, let:

\noindent\begin{minipage}{.5\linewidth}
\begin{align*}
        A_0 &= \left(A[\lceil n/2\rceil-1],\ldots, A[0]\right)_s\\  
B_0 &= \left(B[\lceil n/2\rceil-1],\ldots, B[0]\right)_s
\end{align*}
\end{minipage}%
\begin{minipage}{.5\linewidth}
\begin{align*}
        A_1 &= \left(A[n-1],\ldots, A[\lceil n/2\rceil]\right)_s\\  
        B_1 &= \left(B[n-1],\ldots, B[\lceil n/2\rceil]\right)_s
\end{align*}
\end{minipage}
\vspace{2mm}

\noindent $\textsc{SKIM}$ is then recursively invoked to compute the products $C_0=A_0\times B_0$, $C' = (A_0-A_1)\times (B_1-B_0)$, $C_2= A_1\times B_0$ and $C_3=A_1\times B_1$. Then, let $C_1 = C'+C_0+C_2$. Finally, $C$ is computed as $C=C_0 + s^{n/4}(C_1)+s^{n/2}C_2$. 

The following characterization of the time and memory requirements of \textsc{SKIM} can be obtained by inspection:
\begin{fact}\label{lem:seqcostcok}
Algorithm $\textsc{SKIM}$ competes the product of two $n$-digit integers using at most $16n^{\log_2 3}$ digit wise operations and a memory space of size at most $8n$. 
\end{fact}

Our parallel algorithm $\cok{}$ extends the recursive scheme of $\textsc{SKIM}$ to take advantage of multiple available processors following the general outline discussed in Section~\ref{sec:overview}. We assume that the input factor $n$-digit integers are partitioned in a sequence of processors $\mathbf{P}$ in $n/|\mathbf{P}|$ digits at the beginning of the computation. Further, we assume that each processor is equipped with a local memory of size $M\geq $. Finally, in order to simplify our analysis, we assume $\npr{P}=4\times 3^i$, and $n=\npr{P}\times2^j$, for $i,j\in\mathbb{N}$. If that is not the case, the input integers can be padded with dummy digits, and/or some of the available processors may not be used. The following description and analysis remain correct in these cases with small corrections of the constant factors.

If $n\leq M\npr{P}^{\log_3 2}/10$, then the product is computed by \cok{} in the MI execution mode (i.e., $\cok{}_{MI}$) in $\log_3 3|\mathbf{P}|/4$ \emph{breadth-first} recursive steps. If that is not the case, \cok{} proceeds by executing up to $\log_2 |\mathbf{P}|$ \emph{depth-first} recursive steps. The sub-problems generated after such steps will have input size which allows their solution to be computed in the MI execution mode. \cok{} uses the sequential algorithm $\textsc{SKIM}$ to compute the products of integers $\emph{locally} $. Clearly, any sequential algorithm can be used in place of it. In the following, we present in detail and analyze \cok{}'s execution in the UM and main execution mode.

\subsection{\cok{} in the MI execution mode}
When multiplying input integers with $n$ digits, if $|\mathbf{P}|=1$, $\cok{}_{MI}$ computes the product $C$ using the sequential Karatsuba's algorithm $\textsc{SKIM}$ discussed in the previous section. If $|\mathbf{P}| = 4$, then $\cok{}_{MI}$ proceeds as follows:
\begin{enumerate}
     \item In parallel, $\mathbf{P}[2]$ and $\mathbf{P}[3]$  send to, respectively, $\mathbf{P}[1]$ and $\mathbf{P}[0]$ a copy of the $n/4$ digits of $A(\mathbf{P}[3])$ (resp., $A(\mathbf{P}[2])$.
     \item In parallel, $\mathbf{P}[0]$ and $\mathbf{P}[1]$  send to, respectively, $\mathbf{P}[2]$ and $\mathbf{P}[3]$ a copy of the $n/4$ digits of $B(\mathbf{P}[0])$ (resp., $A(\mathbf{P}[1])$.
     \item After the two previous steps, $A_0$ and a copy of $A_1$ (resp., a copy of $B_0$ and $B_1$) are partitioned in $[\mathbf{P}[1],\mathbf{P}[0]]$ (resp., $[\mathbf{P}[3],\mathbf{P}[2]]$) in $n/4$ digits. In parallel $[\mathbf{P}[1],\mathbf{P}[0]]$ (resp., $[\mathbf{P}[3],\mathbf{P}[2]]$) invoke the parallel subroutine \textsc{DIFF} to compute the flag $A'= |A_0-A-1|$ (resp., $B'=|B_1-B_0|$) and the flag $f_A$ (resp., $f_B$) which equals zero if $A_0=A_1$ (resp., $B_0=B_1$), $1$ if $A_0>A_1$ (resp., $B_1>B_0$), and $-1$ if $A_0<A_1$ (resp., $B_1<B_0$). Once computed $A'$ (resp., $B'$) is partitioned in $[\mathbf{P}[1],\mathbf{P}[0]]$ (resp., $[\mathbf{P}[3],\mathbf{P}[2]]$) in $n/4$ digits, and each processor in the subsequence holds a copy of $f_A$ (resp., $f_B$). Before proceeding processors in $[\mathbf{P}[1],\mathbf{P}[0]]$ (resp., $[\mathbf{P}[3],\mathbf{P}[2]]$) remove the digits of the copy of $A_1$ (resp., $B_0$) from their local memory.
    \item $\mathbf{P}[3]$ sends  $A(\mathbf{P}[3])$, $B(\mathbf{P}[3])$ and $B'(\mathbf{P}[3])$ to $\mathbf{P[2]}$. Then it removes these values, as well as $f_B$, from its local memory. In parallel, $\mathbf{P}[1]$ sends $A(\mathbf{P}[1])$ and $B(\mathbf{P}[1])$ to $\mathbf{P}[0]$, and then it removes them from it local memory. After this step, $\mathbf{P[2]}$ (resp., $\mathbf{P[0]}$) holds $A_1$, $B_1$ and $B'$ (resp., $A_0$ and $B_0$) in its local memory.
    \item $\mathbf{P}[0]$ sends $A'{\mathbf{P}[0]}$ to $\mathbf{P}[1]$, and then removes it and $f_A$ from its local memory.
    \item $\mathbf{P}[2]$ sends $B'$ and $f_B$ to $\mathbf{P}[1]$ and then removes them from its local memory.
    \item In parallel, $\mathbf{P}[0]$ computes $C_0=A_0\times B_0$, $\mathbf{P}[1]$ computes $C'=f_AA'\times f_BB'$, and $\mathbf{P}[2]$ computes $C_2=A_2\times B_2$  using the sequential Karatsuba's algorithm \textsc{SKIM};
    \item In parallel, $\mathbf{P}[0]$ sends to $\mathbf{P}[1]$ a copy of $C_0$, and $\mathbf{P}[2]$ sends to $\mathbf{P}[3]$ a copy of $C_2$.
    \item In parallel, $\mathbf{P}[0]$ (resp., $\mathbf{P}[3]$) sends to $\mathbf{P}[2]$ (resp., $\mathbf{P}[1]$) a copy of the $n/2$ most (resp., least) significant digits of $C_0$ (resp.. $C_2$), and then removes them from its local cache.
    \item $\mathbf{P}[1]$ sends to $\mathbf{P}[2]$ a copy of the $n/2$ most significant digits of $C'$ and then removes them from its local cache.
\end{enumerate}
At the end of the previous steps, two copies of $C_0$ are partitioned in each $[\mathbf{P}[1]\mathbf{P}[0]]$ and in $[\mathbf{P}[1]\mathbf{P}[0]]$ in $n/2$ digits, $C'$ is partitioned in $[\mathbf{P}[2]\mathbf{P}[1]]$ in $n/2$ digits, and two copies of are partitioned $C_2$ is partitioned in each $[\mathbf{P}[2]\mathbf{P}[1]]$ and $[\mathbf{P}[3]\mathbf{P}[2]]$. By construction we have $C(P[0])=C_0(\mathbf{P}[0])$. Let $C'_0 = C_0\mod s^{n/2}$ and $C'_2 = C_2\times s^{n/2}$. $C'_0$, $C_0$, $C_1$, $C_2$ and $C'_2$ are integers with at most $3n/2$ digits each partitioned in $[\mathbf{P}[3],\mathbf{P}[2],\mathbf{P}[1]]$ in $n/2$ digits.
The $3n/2$ most significant digits of $C$ correspond to those of the algebraic sum $C'_0+C'+C_0+C_2+C'_2$. If $C'\geq 0$ such sum can be computed using four consecutive invocations of the $\textsc{SUM}$ subroutine discussed in Section~\ref{sec:distsum} using the sequence $[\mathbf{P}[2],\mathbf{P}[1],\mathbf{P}[0]]$. If instead $C'< 0$, the sum $C'_0+C'+C_0+C_2+C'_2$ can be computed using three invocations of the $\textsc{SUM}$ and one of $\textsc{DIFF}$ using the sequence $[\mathbf{P}[2],\mathbf{P}[1],\mathbf{P}[0]]$. The $3n/2$ digits of such sum are partitioned in $P^*$ in $2n/|\mathbf{P}|$ digits, and, hence, $C$ is partitioned in $\mathbf{P}$ in $2n/|\mathbf{P}|$ digits.

At the end of this procedure, the product $C$ is partitioned in $\mathbf{P}$ in $2n/|\mathbf{P}|$ digits. We use the operations described for $\npr{P}=4$ as the base case for the recursive scheme of $\cok{}_{MI}$.\\


If $|\mathbf{P}| > 4$, $\cok{}_{MI}$ proceeds following a \emph{breadth-first} traversal of the recursion tree, similarly to the operations discussed for the case $|\mathbf{P}| = 4$:
\begin{enumerate}
    \item \textbf{Splitting:} Let $A_0$ and $B_0$ (resp., $A_1$ and $B_1$) be defined as~\ref{eq:partitionAB2}. By assumption, they are partitioned in $\mathbf{P'}$ (resp., $\mathbf{P''}$, that is, the first (resp., second) half of the sequence of $\mathbf{P}$ as defined in~\ref{eq:processorpartition2}. 
    \begin{enumerate}
        \item The values $f_A$, $A'$, $f_B$ and $B'$ are computed following steps analogous to those in the description of the base case $|\mathbf{P}|=4$ (1-4). Processors in $\mathbf{P'}$ (resp., $\mathbf{P''}$) operate as those in $[\mathbf{P}[1],\mathbf{P}[0]]$ (resp., $[\mathbf{P}[3],\mathbf{P}[2]]$), all the communications occur between processors of the same index within the sub-sequence.
    \end{enumerate}
 Consider the following subsequences of $\mathbf{P}$: 
\begin{equation*}
    \begin{aligned}
    \mathbf{P_0} = &[\mathbf{P}[(|\mathbf{P}|/2)-1],\mathbf{P}[(|\mathbf{P}|/2)-3],\ldots,\mathbf{P}[5],\mathbf{P}[3],,\mathbf{P}[2],\mathbf{P}[0]];\\
    \mathbf{P_1} = &[\mathbf{P}[|\mathbf{P}| -2],\mathbf{P}[(3|\mathbf{P}|/4 -2], \mathbf{P}[|\mathbf{P}| - 5],\mathbf{P}[3|\mathbf{P}|/4 -5],
    \ldots, \mathbf{P}[3|\mathbf{P}|/4 +1],  \mathbf{P}[|\mathbf{P}|/2 +1],\\ &\mathbf{P}[|\mathbf{P}|/2 -2],\mathbf{P}[(|\mathbf{P}|/4 -2], \mathbf{P}[|\mathbf{P}|/2 - 5],\mathbf{P}[|\mathbf{P}|/4 -5],\ldots, \mathbf{P}[|\mathbf{P}|/4 +1],  \mathbf{P}[1]];\\
    \mathbf{P_2} = &[\mathbf{P}[(|\mathbf{P}| -1],\mathbf{P}[(|\mathbf{P}| -3],\mathbf{P}[(|\mathbf{P}| -4],\ldots,\mathbf{P}[|\mathbf{P}|/2 +3],\mathbf{P}[|\mathbf{P}|/ +2],\mathbf{P}[2|\mathbf{P}|/3]].
\end{aligned}
\end{equation*}
The sub-sequence $\mathbf{P_0}$ (resp., $\mathbf{P_2}$) includes the processors in the first (resp., second) half of $\mathbf{P}$ 
except those of index $i=1,4,7,\ldots,|\mathbf{P}|/2-5,|\mathbf{P}|/2-2$ (resp., $j=i+|\mathbf{P}|/2+1$). The sub-sequence $\mathbf{P_1}$ is assigned all the remaining processors rearranged in a way which will reduce communications in the latter phases. $\cok{}_{MI}$ assigns each subsequence of the available processors to one of the three sub-problems which are recursively invoked to computed the product $A\times B$: 
\begin{itemize}
    \item The product $A_0\times B_0$ is computed by $\mathbf{P_0}$;
    \item The product $(A_0-A_1)\times (B_1-B_0)$ is computed by $\mathbf{P_1}$;
    \item The product $A_1\times B_1$ is computed by $\mathbf{P_2}$;
\end{itemize}
$\cok{}_{MI}$ transfers to each sequence $\mathbf{P_i}$ the input integers for the corresponding sub-problems. This is achieved in the two following parallel communication steps: 
\begin{enumerate}[resume]
    \item In parallel, each processors $\mathbf{P}[i]$  for $i=1,4,\ldots,|\mathbf{P}|/2-5, |\mathbf{P}|/2-2$ sends to $\mathbf{P}[i-1]$ a copy of the $n/(2|\mathbf{P}|)$ least significant digits of $A(\mathbf{P}[i])$, $B(\mathbf{P}[i])$ and either $A'(\mathbf{P}[i])$ if $i< |\mathbf{P}/2|$ or $B'(\mathbf{P}[i])$ if $i\geq |\mathbf{P}/2|$, and then removes them from its local cache. 
    \item In parallel, each processors $\mathbf{P}[i]$  for $i=1,4,\ldots,|\mathbf{P}|/2-5, |\mathbf{P}|/2-2$ sends to $\mathbf{P}[i+1]$ a copy of the $n/(2|\mathbf{P}|)$ most significant digits of $A(\mathbf{P}[i])$, $B(\mathbf{P}[i])$ and either $A'(\mathbf{P}[i])$ if $i< |\mathbf{P}/2|$ or $B'(\mathbf{P}[i])$ if $i\geq |\mathbf{P}/2|$, and then removes them from its local cache. At the end of this step $A_0$, $B_0$ and $A'$ (resp., $A_1$, $B_1$ and $B'$) are partitioned in $\mathbf{P_0}$ (resp., $\mathbf{P_2}$) in $3n/(2|\mathbf{P}|)$ digits. 
    \item In parallel, all processors $\mathbf{P_0}[i]$ for $i\in\{0,1,\ldots, |\mathbf{P}|/3-1\}$,  send the $3n/(2\npr{P})$ digits of  $A'(\mathbf{P_0}[i])$  to $\mathbf{P_1}[i]$, and then remove them from their local memory.
    \item In parallel, all processors $\mathbf{P_2}[i]$ for $i\in\{0,1,\ldots, |\mathbf{P}|/3-1\}$,  send the $3n/(2\npr{P})$ digits of $B'(\mathbf{P_0}[i])$  to $\mathbf{P_1}[i]$, and then remove them from their local memory.
\end{enumerate}
 
 \item \textbf{Recursive multiplication:} The three sub-products are computed in parallel using $\cok{}_{MI}$: $C_0=A_0\times B_0$ is computed by the processors in $\mathbf{P_0}$, $C' = A'\times B'$ is computed by $\mathbf{P_1}$, and $C_2 = A_1\times B_1$ is computed by $\mathbf{P_2}$. At the end of these recursive calls,  each product is partitioned in the subsequence used to compute it in $3n/|\mathbf{P}|$ digits.
 
 \item \textbf{Recomposition:} $\cok{}_{MI}$ combines $C_0$, $C'$ and $C_2$ to obtain the desired $C$. The steps closely follows those discussed for the base case $|\mathbf{P}| = 4$. 
 \begin{enumerate}
     \item In parallel, each processor $\mathbf{P_0}[i]$ (resp., $\mathbf{P_2}[i]$) for $i=0,3,6,\ldots,|\mathbf{P}_0|-3$, sends to $\mathbf{P}[i+1]$ (resp., $\mathbf{P}[i+1+|\mathbf{P}|/2]$) the $n/|\mathbf{P}|$ most significant digits of $C_0(\mathbf{P_0}[i])$ (resp., $C_2(\mathbf{P_2}[i])$) and then removes them from its local memory. 
     \item In parallel, each processor $\mathbf{P_0}[i]$ (resp., $\mathbf{P_2}[i]$) for $i=2,5,8,\ldots,|\mathbf{P}_0|-1$, sends to $\mathbf{P}[i-1]$ (resp., $\mathbf{P}[i-1+|\mathbf{P}|/2]$) the $n/|\mathbf{P}|$ least significant digits of $C_0(\mathbf{P_0}[i])$ (resp., $C_2(\mathbf{P_2}[i])$) and then removes them from its local memory. Note that each of the $\mathbf{P_0}[i]$'s (resp., $\mathbf{P_2}[i]$'s) is communicating is a different processor in $\mathbf{P_1}$. Hence, all such communication can occur in parallel.
     At the end of these two steps, $C_0$ (resp., $C_2$) is partitioned in the first (resp., second) half of the sequence $\mathbf{P}$, denoted as $\mathbf{P'}$ (resp., $\mathbf{P''}$).
     \item In parallel, each processor $\mathbf{P_1}[i]$:
     \begin{itemize}
         \item for $i=0,2,4, \ldots, |\mathbf{P}_1|/2-2$, sends to $\mathbf{P}[|\mathbf{P}|/4+3i/2]$ the $2n/\npr{P}$ least significant digits of $C'(\mathbf{P_1}[i])$;
         \item for $i=1,3,5, \ldots, |\mathbf{P}_1|/2-1$, sends to $\mathbf{P}[|\mathbf{P}|/4+ \lceil 3i/2\rceil]$ the $2n/|\mathbf{P}|$ most significant digits of $C'(\mathbf{P_1}[i])$ and then removes them from its local memory;
    \end{itemize}
    Further, each processor $\mathbf{P_1}[|\mathbf{P}_1|/2+i]$:
    \begin{itemize}
        \item for $i=0,2,4, \ldots, |\mathbf{P}_1|/2-2$, sends to $\mathbf{P}[|\mathbf{P}|/2+3i/2]$ the $2n/\npr{P}$ least significant digits of $C'(\mathbf{P_1}[i])$, and then removes them from its cache;
         \item for $i=1,3,5, \ldots, |\mathbf{P}_1|/2-1$, sends to $\mathbf{P}[|\mathbf{P}|/2+ \lceil 3i/2\rceil]$ the $2n/|\mathbf{P}|$ most significant digits of $C'(\mathbf{P_1}[i])$;
     \end{itemize}
     As in these steps each processor in $\mathbf{P}_1$ communicates with a single, distinct, processor in either $\mathbf{P_0}$ or $\mathbf{P_2}$, all these communications may occur in parallel.
     \item In parallel, each  processor $\mathbf{P'}[i]$ (resp., $\mathbf{P''}[i+|\mathbf{P}|/4]$) for $i=0,2,\ldots,|\mathbf{P}|/4-2$ sends to $\mathbf{P'}[i+|\mathbf{P}|/4]$ (reps., $\mathbf{P''}[i]$) a copy of the $2n/\npr{P}$ digits of $C_0(\mathbf{P'}[i])$ (resp., $C_2(\mathbf{P''}[i+|\mathbf{P}|/4])$).
     \item In parallel, each  processor $\mathbf{P'}[i+|\mathbf{P}|/4]$, for $i=0,1,\ldots,|\mathbf{P}|/4-1$, sends to $\mathbf{P''}[i]$ a copy of the digits og $2n/\npr{P}$ of $C_0(\mathbf{P'}[i])$.
     \item In parallel, each  processor $\mathbf{P''}[i]$, for $i=0,1,\ldots,|\mathbf{P}|/4-2$, sends to $\mathbf{P'}[i+|\mathbf{P}|/4]$  a copy of the $2n/\npr{P}$ of digits of $C_2(\mathbf{P'}[i])$.
 \end{enumerate} 
 \vspace{-2mm}
 At the end of these operations, $C_0$ is partitioned in $\mathbf{P}'$ in $2n/|\mathbf{P}|$ digits, $C_2$ is partitioned in $\mathbf{P}''$ in $2n/|\mathbf{P}|$ digits. Further, $C'$ and copies of $C_0$ and $C_2$ are partitioned in the subsequence of processors $[\mathbb{P}[3|\mathbf{P}|/4-1],\ldots,\mathbb{P}[|\mathbf{P}|/4-1]]$ in $2n/|\mathbf{P}|$ digits. The computation of $C$ is completed following steps analogous to those discussed for the base case $|\mathbf{P}|=4$.
\end{enumerate}
\vspace{-3mm}
\begin{theorem}\label{thm:costcokUM}
Let $A$ and $B$ be two $n$-digit integers portioned in a sequence of processors $\mathbf{P}$, with $n\geq |\mathbf{P}|$, in $n/|\mathbf{P}|$ digits. $\cok{}_{MI}$ computes the product $C=A\times B$ using the processors in $\mathbf{P}$, provided that each processor is equipped with a local memory of size at least $M_{\cok{}_{MI}}(n,|\mathbf{P}|)=10n/\npr{P}^{\log_3 2}$.  We have:
\vspace{-4mm}
\begin{align*}
    T_{\cok{}_{MI}}\left(n,|\mathbf{P}|\right) &\leq 173\frac{n^{\log_2 3}}{\npr{P}}\\
    BW_{\cok{}_{MI}}\left(n,|\mathbf{P}|\right) &\leq 174\frac{n}{\npr{P}^{\log_3 2}}\\
    L_{\cok{}_{MI}}\left(n,|\mathbf{P}|\right) &\leq 25\log^2_2 \npr{P}
\end{align*}
\end{theorem}
\begin{proof}
$\cok{}_{MI}$ correctly computes $C=A\times B$ by inspection.

If $|\mathbf{P}|=1$, the product $C=A\times B$ is computed locally by the single available processor using algorithm $\textsc{SKIM}$. By Lemma~\ref{lem:seqcostlim} we thus have $M_{\cok{}_{MI}}(n,1)\le 8n$,  $T_{\cok{}_{MI}}(n,1)\le 8n^2$, $BW_{\cok{}_{MI}}(n,1)=0$ and $L_{\cok{}_{MI}}(n,1)=0$.

For $|\mathbf{P}|=4$, the product $C=A\times B$ is computed using three of the four available processors. By the description, at any time each processor may need to maintain in its local memory at most $3n/2$ digits of the input, $4\times n/2$ digits of the output of the sub-problems, the data required to invoke $\textsc{DIFF}$ to compute $A'$ and $B'$, the data required to execute $\textsc{SKIM}$ for input integers of size $n/2$, and the data required to invoke $\textsc{SUM}$ and $\textsc{DIFF}$ in the final recombination step. Note that, once the outputs of the subproblems are computed, it is unnecessary to maintain the respective input values in the memory. By Lemma~\ref{lem:costdiff} and Lemma~\ref{lem:costsum}, $4\frac{n}{2}+5\geq M_{\textsc{DIFF}}(3n/2,3)\geq M_{\textsc{SUM}}(3n/2,3)$. 
\vspace{-3mm}
\begin{align}
    M_{\cok{}_{MI}}\left(n,4\right) &\leq \mymax{\frac{3n}{4}+ M_{\textsc{DIFF}}\left(\frac{n}{2},2\right), M_{\textsc{SKIM}}\left(\frac{n}{2}\right), \frac{3n}{4} +  M_{\textsc{DIFF}}\left(\frac{3n}{2},3\right)}\nonumber\\
    &\leq \mymax{\frac{3n}{4} + 4\frac{n}{4}+5, 4n, \frac{3n}{2}+4\frac{n}{2} + 5}\nonumber\\
    &< 4n
\end{align}
where the last passage follows as, by assumption, $ n\geq \npr{P}=4$.
The computation time is given by the sum of the time required for computing the differences in step (4), the invocation of $\textsc{SKIM}$ in step (8), and for computing a difference and three sums, as discussed at the end of the presentation of the base case:
\vspace{-2mm}
\begin{align}
    T_{\cok{}_{MI}}\left(n,4\right) &\leq T_{\textsc{DIFF}}\left(\frac{n}{2},2\right) + T_{\textsc{SKIM}}\left(\frac{n}{2}\right) + T_{\textsc{DIFF}}\left(\frac{3n}{2},3\right) + 3T_{\textsc{SUM}}\left(\frac{3n}{2},2\right)\nonumber\\
    &< 7\frac{n}{4}+16\left(\frac{n}{2}\right)^{\log_2 3} + 7\frac{n}{2}+3\times 6\frac{n}{2}\nonumber\\
    &< 12n^{\log_2 3}
\end{align}
where the last passage follows as, by assumption, $ n\geq \npr{P}=4$.
Similarly, the \io{} cost of $\cok{}_{MI}$ for the base case is bounded by the cost of operations discussed in the description of the base case as follows:
\begin{align}
    BW_{\cok{}_{MI}}\left(n,4\right) &\leq \frac{n}{4}+ \frac{n}{4}+ BW_{\textsc{DIFF}}\left(\frac{n}{2},2\right) +\frac{3n}{4} +\frac{n}{4}+\frac{n}{4}+1+ BW_{\textsc{SKIM}}\left(\frac{n}{2}\right) +n\nonumber\\
    & +\frac{n}{2}+\frac{n}{2} + BW_{\textsc{DIFF}}(n,2) +3 BW_{\textsc{SUM}}(n,2)\nonumber \\
    &= 4n + BW_{\textsc{DIFF}}\left(\frac{n}{2},2\right) + BW_{\textsc{DIFF}}(n,2) +3 BW_{\textsc{SUM}}(n,2)\nonumber\\
    &= 4n + 5+5+ 3\times 4 +1 \label{eq:cokp4bw1}\\
    &< 10n\label{eq:cokp4bw2}\\
    L_{\cok{}_{MI}}\left(n,4\right) &\leq 1+1+ L_{\textsc{DIFF}}\left(\frac{n}{2},2\right) +1+1+1  + L_{\textsc{SKIM}}\left(\frac{n}{2}\right) +1+1+1\nonumber\\
    &+L_{\textsc{DIFF}}(n,2) +3 L_{\textsc{SUM}}(n,2)\nonumber \\
    &=8 +L_{\textsc{DIFF}}\left(\frac{n}{2},2\right)+L_{\textsc{DIFF}}(n,2) +3 L_{\textsc{SUM}}(n,2)\nonumber\\
    &=8+3+3+3\times2 \label{eq:cokp4bw3}\\
    &=20\nonumber
\end{align}
where~\eqref{eq:cokp4bw1} and ~\eqref{eq:cokp4bw3} follow from Lemma~\ref{lem:costsum} and Lemma~\ref{lem:costdiff}, and~\eqref{eq:cokp4bw2} follows from the assumption $n\geq|\mathbf{P}|=4$. This concludes the analysis of the base case.

In the following, we assume $|\mathbf{P}|>4$. $\cok{}_{MI}$ proceeds to execute $\log_3 \npr{P}/4$ \emph{breadth-first} recursion steps. In each recursion branch, the available processors are divided into three disjoint subsequences and assigned to a different recursive invocation of $\cok{}_{MI}$. After $\log_3 |\mathbf{P}|/4$ recursive breadth-first levels, each of the generated sub-problems is assigned to a subsequence of four processors and $\cok{}_{MI}$ proceeds according to the base case.
The analysis largely follows that for the base case $|P|=4$ with only minor adjustments corresponding to the minor difference in the use of the available processors:
\vspace{-3mm}
\begin{align}
    M_{\cok{}_{MI}}\left(n,|\mathbf{P}|\right) &\leq \max \Big\{\frac{3n}{|\mathbf{P}|}+ M_{\textsc{DIFF}}\left(\frac{n}{2},\frac{\npr{P}}{2}\right), M_{\cok{}_{MI}}\left(\frac{n}{2},\frac{|\mathbf{P}|}{3}\right),\nonumber\\ &\qquad \qquad 3\frac{2n}{\npr{P}} + M_{\textsc{DIFF}}\left(n,\frac{\npr{P}}{2},\right)\Big\}\nonumber\\
    &\leq \mymax{3\frac{n}{|\mathbf{P}|}+4\frac{n}{|\mathbf{P}|}+5,  M_{\cok{}_{MI}}\left(\frac{n}{2},\frac{|\mathbf{P}|}{3}\right),6\frac{n}{\npr{P}}+4\frac{2n}{\npr{P}}+5}\nonumber\\
    &\leq \mymax{M_{\cok{}_{MI}}\left(\frac{n}{2},\frac{|\mathbf{P}|}{3}\right),14\frac{n}{\npr{P}}+5}\nonumber\\
    &\leq \mymaxo{i=1,\ldots,\log_3 \npr{P}/4}{M_{\cok{}_{MI}}\left(\frac{n}{2^i},\frac{|\mathbf{P}|}{3^i}\right),14\frac{n}{\npr{P}}\left(\frac{3}{2}\right)^{i-1}+5}\nonumber\\
    &\leq \mymax{M_{\cok{}_{MI}}\left(n\left(\frac{4}{\npr{P}}\right)^{\log_3 2},4\right),14\frac{n}{\npr{P}}\left(\frac{3}{2}\right)^{\log_3 \npr{P}/4 -1}+5}\nonumber\\
    &< \mymax{\frac{10n}{\npr{P}^{\log_3 2}},6\frac{n}{\npr{P}^{\log_2 3}}+5}\nonumber\\
    &= 10n/\npr{P}^{\log_3 2}.\nonumber
\end{align}

Hence, by Lemma~\ref{lem:seqcostlim}, $M_{\cok{}_{MI}}(n,|\mathbf{P}|)\geq 8n/\sqrt{|\mathbf{P}|}$. We can thus conclude that all such requirement are met if each processor is equipped with a local memory of size $M_{\cok{}_{MI}}(n,|\mathbf{P}|)=12n/\sqrt{|\mathbf{P}|}$.\\

By construction, in each recursion level of the MI execution mode, $\cok_{MI}$ computes the differences $A'$,$B'$, and then it recursively invokes itself on three distinct subsequences of available processors to compute the three generated subproblems. Then, $\cok_{MI}$ combines the output of such subproblems to conclude the computation of $C$. By Lemma~\ref{lem:costsum} and Lemma~\ref{lem:costdiff}, we have:
\vspace{-3mm}
\begin{align*}
    T_{\cok{}_{MI}}\left(n,\npr{P}\right) &\leq T_{\textsc{DIFF}}\left(\frac{n}{2},\frac{\npr{P}}{2}\right) + T_{\cok{}_{MI}}\left(\frac{n}{2},\frac{|\mathbf{P}|}{3}\right)\\ &\quad + T_{\textsc{DIFF}}\left(\frac{3n}{2},\frac{3\npr{P}}{4}\right) + 3 T_{\textsc{SUM}}\left(\frac{3n}{2},3\npr{P}/4\right)\nonumber\\
    &< 7\frac{n}{\npr{P}}+5 \log_2 \frac{\npr{P}}{2}+T_{\cok{}_{MI}}\left(\frac{n}{2},\frac{|\mathbf{P}|}{3}\right)+ 14\frac{n}{\npr{P}}+5 \log_2 \frac{3\npr{P}}{4}\\ &\quad + 3\left(12\frac{n}{\npr{P}}+4 \log_2 \frac{3\npr{P}}{4}\right)\nonumber\\
    &< 57\frac{n}{\npr{P}} + 22 \log_2 \npr{P}+T_{\cok{}_{MI}}\left(\frac{n}{2},\frac{|\mathbf{P}|}{3}\right)\nonumber\\
    &\leq T_{\cok{}_{MI}}\left(n\left(\frac{4}{\npr{P}}\right)^{\log_3 2},4\right) + 57\frac{n}{\npr{P}}\sum_{i=0}^{\log_3 \npr{P}/4-1} \left(\frac{3}{2}\right)^{i}\nonumber\\&\quad + 22\sum_{i=0}^{\log_3 \npr{P}/4-1} \log_2 \frac{\npr{P}}{3^i}\nonumber\\
    &< 48\frac{n^{\log_2 3}}{\npr{P}}+69\frac{n}{\npr{P}^{\log_3 2}}+ 14\log^2_2 \npr{P}\nonumber\\
    &<117\frac{n^{\log_2 3}}{\npr{P}}+56\frac{n^{\log_2 3}}{\npr{P}}\nonumber
\end{align*}
where the last passage follows from the fact that, by assumption, $n>\npr{P}\geq 2$, and that for such values $4\frac{n^{\log_2 3}}{\npr{P}}>\log^2_2\npr{P}$.
The \io{} operation of $\cok{}_{MI}$ have been presented in detail in the description of the algorithm. Here we present the analysis of the recursion, which allows us to characterize both bandwidth and latency cost. To aid readability, in the first line of the analysis of the bandwidth (resp., latency), each term of the sum corresponds to the bandwidth (resp., latency) cost of each phase specified in the description of $\cok{}_{MI}$.  By Lemma~\ref{lem:costsum} and Lemma~\ref{lem:costdiff}, we have:
\begin{align}
    BW_{\cok{}_{MI}}\left(n,\npr{P}\right) 
    &\leq \frac{n}{\npr{P}}+ \frac{n}{\npr{P}}+ BW_{\textsc{DIFF}}(n/2,\npr{P}/2,n/\npr{P}) +\frac{n}{\npr{P}} +\frac{n}{\npr{P}}+2\times \frac{3n}{2\npr{P}}+2\nonumber\\ 
    &\quad+ BW_{\cok{}_{MI}}\left(\frac{n}{2},\frac{|\mathbf{P}|}{3}\right) +2\times\frac{n}{\npr{P}} \nonumber\\
    &\quad +4\times\frac{2n}{\npr{P}} + BW_{\textsc{DIFF}}\left(\frac{3n}{2},\frac{3\npr{P}}{4}\right) +3 BW_{\textsc{SUM}}\left(\frac{3n}{2},\frac{3|\mathbf{P}|}{4}\right)\nonumber \\
    &\leq  5\log_2 \frac{\npr{P}}{2}+ BW_{\cok{}_{MI}}\left(\frac{n}{2},\frac{|\mathbf{P}|}{3}\right) + 5\log_2 \frac{3\npr{P}}{4} +3\times 4\log_2 \frac{3\npr{P}}{4}\nonumber\\&\quad + 20\frac{n}{\npr{P}}+2\nonumber \\
    &\leq  BW_{\cok{}_{MI}}\left(\frac{n}{2},\frac{|\mathbf{P}|}{3}\right)+22\log_2 \npr{P} + 20\frac{n}{\npr{P}}\nonumber  \\
    &\leq BW_{\cok{}_{MI}}\left(n\left(\frac{4}{\npr{P}}\right)^{\log_3 2},4\right) + 20\frac{n}{\npr{P}}\sum_{i=0}^{\log_3 \npr{P}/4-1} \left(\frac{3}{2}\right)^{i}\nonumber\\&\quad + 22\sum_{i=0}^{\log_3 \npr{P}/4-1} \log_2 \frac{\npr{P}}{3^i}\nonumber\\
    &< 24\frac{n}{\npr{P}^{\log_3 2}}+24\frac{n}{\npr{P}^{\log_3 2}}+ 14\log^2_2 \npr{P}.\nonumber\\
    &< 48\frac{n}{\npr{P}^{\log_3 2}}+ 126\frac{n}{\npr{P}^{\log_3 2}}.\label{eq:bwcokum}
\end{align}
where~\eqref{eq:bwcokum} follows from the fact that, by assumption, $n>\npr{P}\geq 2$, and that for such values $9\frac{n}{\npr{P}^{\log_3 2}}>\log^2_2\npr{P}$. Finally,
\begin{align}
    L_{\cok{}_{MI}}\left(n,\npr{P}\right) 
    &\leq 1+1+ L_{\textsc{DIFF}}\left(\frac{n}{2},\frac{|\mathbf{P}|}{2}\right) +1+1+1+1 + L_{\cok{}_{MI}}\left(\frac{n}{2},\frac{|\mathbf{P}|}{3}\right) \nonumber\\
    &\quad+2 +4 + L_{\textsc{DIFF}}\left(\frac{3n}{2},\frac{3|\mathbf{P}|}{4}\right) +3 L_{\textsc{SUM}}\left(\frac{3n}{2},\frac{3|\mathbf{P}|}{4}\right)\nonumber \\
    &\leq  3\log_2 \frac{\npr{P}}{2}+ L_{\cok{}_{MI}}\left(\frac{n}{2},\frac{|\mathbf{P}|}{3}\right) + 3\log_2 \frac{3\npr{P}}{4} +3\times 2\log_2 \frac{3\npr{P}}{4} + 12 \nonumber\\
    &\leq  L_{\cok{}_{MI}}\left(\frac{n}{2},\frac{|\mathbf{P}|}{3}\right)+12\log_2 \npr{P}\nonumber\\
    &\leq L_{\cok{}_{MI}}\left(n\left(\frac{4}{\npr{P}}\right)^{\log_3 2},4\right) + \sum_{i=0}^{\log_3 \npr{P}/4-1} 12\log_2 \npr{P}\nonumber\\
    &< 20+ 5\log^2_2 \npr{P}\nonumber\\
    &<25\log^2_2 \npr{P}\nonumber
\end{align}
\end{proof}

\subsection{\cok{} in the main execution mode}
An important consequence of of Theorem~\ref{thm:costcoimUM} is that to execute $\cok{}_{MI}$ to multiply two $n$-digit integers each processor must be equipped with a memory of size at least $10n/|\mathbf{P}|^{\log_3 2}$. Considering the aggregate memory available in $|\mathbf{P}|$ processors, this implies that $\cok{}_{MI}$ can only be used to multiply $n$-digit integers where $n\in \BO{M|\mathbf{P}|^{\log_3 2}}$.  

As outlined in the general framework in Section~\ref{sec:overview}, in this section, we describe how to combine the MI execution mode in $\cok{}_{MI}$ with a \emph{depth-first} scheduling of the recursively generated subproblems. Our algorithm, \cok{}, allows to compute the product of two $n$-digit integers provided that (i) $M\geq 40n/|\mathbf{P}|$ and (ii) $M\geq \log_2 \npr{P}$. By (i), \cok{} can thus be used to multiply $n$-digit integers where $n\in \BO{M|\mathbf{P}|}$. That is, its memory requirement corresponds asymptotically to the amount of overall memory required to store the input integers and/or the product.

In its main execution mode, \cok{} executes up to $\log_2 \frac{20n}{M\npr{P}^{\log_3 2}}$ of \emph{depth-first} steps, until the size of the subproblems being generated is such that they can be computed by the sequence of processors $\mathbf{P}$ using \cok{} according to the steps discussed for the MI execution mode.

When $\cok{}\left(n,A,B,\mathbf{P},m\right)$ is invoked, if $n\leq 12n\sqrt{|\mathbf{P}|}$, the product is computed by invoking $\cok{}_{MI}$. Otherwise \cok{} proceeds as follows:
\begin{enumerate}
    \item In parallel, each processor $\mathbf{P}[i]$, for $i=0,1,\ldots,|\mathbf{P}|/2-1$, sends to $\mathbf{P}[i+|\mathbf{P}|/2]$  a copy of the $n/(2|\mathbf{P}|)$ most significant digits of $A(\mathbf{P}[i])$ and $B(\mathbf{P}[i])$, and then remove them from their cache.
    \item In parallel, each processor $\mathbf{P}[|\mathbf{P}|/2+i]$, for $i=0,1,\ldots,|\mathbf{P}|/2-1$, sends to $\mathbf{P}[i]$ a copy of the $n/(2|\mathbf{P}|)$ least significant digits of $A(\mathbf{P}[|\mathbf{P}|/2+i])$ and $B(\mathbf{P}[|\mathbf{P}|/2+i])$), and then remove them from their cache.
\end{enumerate}
At the end of step (1) (resp., (2)), $A_0$ and $B_0$ (resp., $A_1$ and $B_1$)  are partitioned in the sequence:
    $$\mathbf{\tilde{P}}= [\mathbf{P}[|\mathbf{P}|-1], \mathbf{P}[|\mathbf{P}|/2-2], \ldots,\mathbf{P}[|\mathbf{P}|/2],\mathbf{P}[0]$$
    in $n/(2|\mathbf{P}|)$ digits. The even (resp., odd) index processors in the sequence $\mathbf{\tilde{P}}$ are the processors in the first (resp., second) half of the sequence $\mathbf{P}$).
\begin{enumerate}[resume]
    \item \cok{} is then recursively invoked to compute $C_0= A_0\times B_0$ using the sequence $\mathbf{\tilde{P}}$. Once computed, $C_0$ is partitioned in $\mathbf{\tilde{P}}$ in $n/|\mathbf{P}|$ digits.
    \item \cok{} is then recursively invoked to compute $C_2= A_1\times B_1$ using the sequence $\mathbf{\tilde{P}}$. Once computed, $C_2$ is partitioned in $\mathbf{\tilde{P}}$ in $n/|\mathbf{P}|$ digits.
    \item \textsc{DIFF} is invoked on $\mathbf{\tilde{P}}$ to compute $A'=|A_0-A_1|$ and the flag $f_A$ such that $f_A=0$ if $A_0 = A_1$, $f_A=1$ if $A_0 > A_1$, and $f_A=-1$ if $A_0 < A_1$. Then each processor removes the digits of $A_0$ and $A_1$ from its local memory.
    \item \textsc{DIFF} is invoked on $\mathbf{\tilde{P}}$ to compute $B'=|B_1-A_0|$ and the flag $f_B$ such that $f_B=0$ if $B_0 = B_1$, $f_B=1$ if $B_0 < B_1$, and $f_B=-1$ if $B_0 < B_1$. Then each processor removes the digits of $B_0$ and $B_1$ from its local memory.
    \item If $f_A\times f_B = 0$, each processor in $\mathbf{\tilde{P}}$ sets $C'(\mathbf{\tilde{P}})=0$. Otherwise, \cok{} is then recursively invoked to compute $C'= A'\times B'$ using the sequence $\mathbf{\tilde{P}}$. Once computed, $C'$ is partitioned in $\mathbf{\tilde{P}}$ in $n/|\mathbf{P}|$ digits.
\end{enumerate}
To complete the computation of $C$, \cok{} opportunely redistributes and combine $C_0$, $C'$ and $C_2$
\begin{enumerate}[resume]
    \item \textsc{SUM} is invoked on $\mathbf{\tilde{P}}$ to compute the sum $C_0+C_2$, which, once computed is partitioned in $\mathbf{\tilde{P}}$ in $n/|\mathbf{P}|$ digits.
    \item If $f_A\times f_B = 1$ (resp., $f_A\times f_B = -1$), \textsc{SUM} (resp., \textsc{DIFF}) is invoked on $\mathbf{\tilde{P}}$ to compute $C_1 =(f_A\times f_B)C'+C_0+C_2$, which, once computed is partitioned in $\mathbf{\tilde{P}}$ in $n/|\mathbf{P}|$ digits.
    \item In parallel each processor $\mathbf{\tilde{P}}[|\mathbf{P}|/2+i]$, for $i = 0,1,\ldots,|\mathbf{P}|/2-1$, sends to $\mathbf{\tilde{P}}[i]$ a copy of the $n/|\mathbf{P}|)=$ digits of $C_0(\mathbf{\tilde{P}}[|\mathbf{P}|/2+i])$, henceforth referred as $C'_0(\mathbf{\tilde{P}}[|\mathbf{P}|/2+i])$, and then it removes them from its local memory. Let $C'_0$ be the integer value whose digits correspond to the $n/2$ most significant digits of $C_0$ (i.e., $C'_0= \lceil C_0/s^{n/2}\rceil$). $C'_0$ is partitioned in $[\mathbf{\tilde{P}}[\npr{P}/2-1],\ldots,\mathbf{\tilde{P}}[0]]$.
    \item In parallel each processor $\mathbf{\tilde{P}}[i]$, for $i = 0,1,\ldots,|\mathbf{P}|/2-1$, sends to $\mathbf{\tilde{P}}[i+\npr{P}/2]$ a copy of the $n/|\mathbf{P}|$ digits of $C_2(\mathbf{\tilde{P}}[|\mathbf{P}|/2+i])$, henceforth referred as $C'_2(\mathbf{\tilde{P}}[|\mathbf{P}|/2+i])$, and then removes them from its local memory. Let $C'_2$ be the integer value whose digits correspond to the $n/2$ least significant digits of $C_2$ (i.e., $C'_2=  C_2 \mod s^{n/2}$). $C'_2$ is partitioned in $[\mathbf{\tilde{P}}[\npr{P}-1],\ldots,\mathbf{\tilde{P}}[\npr{P}/2]]$.
    \item Let $C^* = C'_2\times s^{n/2}+C'_0$. By construction, $C^*$ is partitioned in $\mathbf{\tilde{P}}$ in $n/\npr{P}$ digits. $C'_1=C^*+C_1$ is computed using the parallel subroutine \textsc{SUM} on $\mathbf{\tilde{P}}$.
\end{enumerate}
    Note that the values $C_0+C_2$, $C_1$ and $C'_1$ may have up to $n+\lceil 3/s\rceil$ non-zero digits in their base-$s$ expansion, with the most significant(s) $\lceil 3/s\rceil\leq 2$ are held in the memory of $\mathbf{\tilde{P}}[\npr{P}-1]$. Sums and differences involving these values can still be computed with $\textsc{SUM}$ and $\textsc{DIFF}$ by increasing by $\lceil 3/s\rceil\leq 2$ their memory requirement and computation time. 
    By construction, $n$ least significant digits of $C'_1$ correspond to the $n$ digits of $C$ from the $(n/2-1)$-th least significant, to the $(3n/2-1)$-th least significant. The $n/2$ least significant digits of $C_0$ correspond to the $n/2$ least significant digits of $C$. Finally, the digits in the base-$s$ expansion of $C^m = \lfloor C_2/s^{n/2}\rfloor + \lfloor C'_1/s^{n}\rfloor$ correspond to the $n/2$ most significant digits of $C$.  
\begin{enumerate}[resume]
    \item Let $d = \lceil C'_1/s^{n}\rceil$. $\mathbf{\tilde{P}}[\npr{P}-1]$ sends $d$ to $\mathbf{\tilde{P}}[\npr{P}/2]$ and then removes its from its cache. 
    \item Let $C''_2 = \lfloor C_2/s^{n/2}\rfloor$. By construction, $C''_2$ is partitioned in $[\mathbf{\tilde{P}}[\npr{P}-1],\ldots,\mathbf{\tilde{P}}[\npr{P}/2]]$ in $n/\npr{P}$ digits. \textsc{SUM} is invoked on this subsequence to compute $C^m=C''_2+d$. As previously mentioned, its base-$s$ expansion correspond to the $n/2$ most significant digits of $C$.
    \item In parallel each processor $\mathbf{\tilde{P}}[i]$ (resp., $\mathbf{\tilde{P}}[i-1+\npr{P}/2]$) for $i = 1,3,\ldots,|\mathbf{P}|/2-1$, sends to $\mathbf{\tilde{P}}[i-1]$ (resp., $\mathbf{\tilde{P}}[i+\npr{P}/2]$) the $n/|\mathbf{P}|$ digits of $C_0(\mathbf{\tilde{P}}[i])$ (resp., $C^m(\mathbf{\tilde{P}}[i])$), and then removes them from its local memory. After these operations, $C_0\mod s^{n/2}$ (resp., $C^m$)  is partitioned in $[\mathbf{P}[\npr{P}/4-1],\ldots,\mathbf[0]]$ in $2n/|\mathbf{P}|$ (resp., $[\mathbf{P}[\npr{P}-1],\ldots,\mathbf[3\npr{P}/4]]$ in $2n/|\mathbf{P}|$)digits.
    \item In parallel each processor $\mathbf{\tilde{P}}[i]$ (resp., $\mathbf{\tilde{P}}[i+|\mathbf{P}|/2-1]$), for $i=1,3,5,\ldots,|\mathbf{P}|/2-1$, sends to $\mathbf{\tilde{P}}[i-1]$ (resp., $\mathbf{\tilde{P}}[i+|\mathbf{P}|/2]$) the $n/\npr{P}$ digits of $C'(\mathbf{\tilde{P}}[i])$ (resp., $C'(\mathbf{\tilde{P}}[i+\npr{P}/2-1])$), and then removes them from its cache.
    \item In parallel each processor $\mathbf{\tilde{P}}[i]$ (resp., $\mathbf{\tilde{P}}[i+|\mathbf{P}|/2]$), for $i=0,2,4,\ldots,|\mathbf{P'}|/2-2$, sends to $\mathbf{\tilde{P}}[i+|\mathbf{P}|/2]$ (resp., $\mathbf{\tilde{P}}[i+1]$) the $n/\npr{P}$ digits of $C'(\mathbf{\tilde{P}}[i])$ (resp., $C'(\mathbf{\tilde{P}}[i+|\mathbf{P}|/2-1])$), and then remove them from its cache. After these last two steps, $C'$ is partitioned in $[\mathbf{P}[3\npr{P}/4-1],\ldots,\mathbf[\npr{P}/4]]$ in $2n/\npr{P}$ digits.
\end{enumerate}

The following theorem rigorously characterizes the performance of \cok{}:

\begin{theorem}\label{thm:cok}
Let $A$ and $B$ be two $n$-digit integers portioned in a sequence of processors $\mathbf{P}$, with $n\geq |\mathbf{P}|$, in $n/|\mathbf{P}|$ digits. $\cok{}$ computes the product $C=A\times B$ using the processors in $\mathbf{P}$, provided that each processor is equipped with a local memory of size at least $M_{\cok{}}(n,|\mathbf{P}|)\geq \max\{40n/|\mathbf{P}|,\log_2 \npr{P}\}$.  We have:
\begin{align*}
    T_{\cok{}}\left(n,|\mathbf{P}|, M\right) &\leq 675\frac{n^{\log_2 3}}{|\mathbf{P}|}\\
    BW_{\cok{}}\left(n,|\mathbf{P}|, M\right) &\leq  1708\left(\frac{n}{M}\right)^{\log_23}\frac{M}{\npr{P}}\\
    L_{\cok{}}\left(n,|\mathbf{P}|, M\right) &\leq 8728\frac{n^{\log_23}}{\npr{P}M^{\log_2 3}}\log_2^2|\mathbf{P}|
\end{align*}
\end{theorem}
\begin{proof}
The algorithm is correct by inspection. If $n\leq M|\mathbf{P}|^{\log_3 2}/10$, the statement follows from Theorem~\ref{thm:costcoimUM}. In the following we assume $M|\mathbf{P}|/24\geq n>M\npr{P}^{\log_3 2}/12$. 
\cok{} then proceeds by executing $\ell$ consecutive depth-first steps, where $0< \ell\leq \lceil \log_2 n/(40M\sqrt{P})\rceil$. In each recursion step, in addition to the space required for the recursive invocation to $\cok{}$, each processor must maintain $2\times n/(|\mathbf{P}|)$ digits for the input of the problem being computed, at most $4n/|\mathbf{P}|+2$ (including the temporary value $C'_1$) digits of the outputs of the recursive subproblems, and the space required for the invocation of $\textsc{DIFF}$ and $\textbf{SUM}$ used to combine the outputs of the subproblems. Note that the recursive calls to $\cok{}$ and the invocations of $\textsc{DIFF}$ and $\textbf{SUM}$ are always executed in distinct steps of the the same recursion level. Hence the memory space used for each can be \emph{reused}. Further, by Lemma~\ref{lem:costdiff} and Lemma~\ref{lem:costsum}, the memory requirement of $\textsc{DIFF}$ is higher of that on $\textbf{SUM}$ for same input size and number of available processors, we have:
\begin{align}
        M_{\cok{}}\left(n,\npr{P}\right) &\leq \frac{4n}{\npr{P}}+2 + \mymax{M_{\cok}\left(n,\npr{P}\right), M_{DIFF} \left(n, |\mathbf{P}|\right)}\nonumber\\
        &\leq \frac{4n}{\npr{P}}+2 + \mymax{M_{\cok}\left(n,\npr{P}\right), \frac{4n}{\npr{P}}+5}\label{eq:cokmem1}\\
        &\leq \frac{4n}{\npr{P}}+2 + M_{\cok{}}\left(n,\npr{P}\right)\nonumber\\
        &\leq \frac{4n}{\npr{P}}\sum_{i=0}^{\lceil\log_2 \frac{20n}{M\npr{P}^{\log_3 2}}\rceil-1}2^{-i}+2\left(\lceil\log_2 \frac{20}{M\npr{P}^{\log_3 2}}\rceil-1\right)\nonumber\\&\quad+ M_{\cok{}}\left(\frac{M\npr{P}^{\log_3 2}}{20},\npr{P},\frac{M\npr{P}^{\log_3 2}}{20\npr{P}^{1-\log_32}}\right)\nonumber\\
        &\leq \frac{8n}{\npr{P}}+2\log_2 \frac{20n}{M\npr{P}^{\log_3 2}}+ M_{\cok{}_{MI}}\left(\frac{M\npr{P}^{\log_3 2}}{20},\npr{P},\frac{M\npr{P}^{\log_3 2}}{20\npr{P}^{1-\log_32}}\right)\nonumber\\
        &\leq 20\frac{n}{\npr{P}}+\frac{M}{2}\label{eq:cokmem2}
\end{align}
where~\eqref{eq:cokmem1} follows from Lemma~\ref{lem:costdiff}, and   \eqref{eq:cokmem2} 
from Theorem~\ref{thm:costcokUM}. 
By construction, after $\ell$ depth-first steps, the available memory is reduced by at most $10n/|\mathbf{P}|$. As, by assumption, $M\geq 40 n/|\mathbf{P}|$, at least half of the space $M$ originally assigned to each processor is still available. After at most $\lceil \log_2 20n/(M\npr{P}^{\log_3 2})\rceil$ recursive steps, the generated sub-problems will have size at most $M\npr{P}^{\log_3 2}/20$. Thus, by Theorem~\ref{thm:costcokUM}, can be computed using $\cok{}$ using at most $M/2$ memory locations. This concludes the proof of the memory requirement for \cok{}.

The computation time required in a depth-first recursion level of \cok{}'s execution is bounded by the time required for evaluating $A'$ and $B'$, the computation steps required for the consecutive recursive invocations of \cok{} on the generated subproblems, plus the time required to combine the outputs of the sub-problems to compute the product itself using $\textsc{DIFF}$ and three invocations of $\textsc{SUM}$. By following the description of the algorithm, we have:
\begin{align*}
    T_{\cok{}}\left(n, |\mathbf{P}|,M\right) 
    &\leq T_{\textsc{DIFF}}\left(\frac{n}{2},\frac{\npr{P}}{2}\right)+3 T_{\cok{}}\left(\frac{n}{2}, |\mathbf{P}|,M\right)+ T_{\textsc{SUM}}\left(n,\npr{P}\right)\\&\quad+ T_{\textsc{DIFF}}\left(n,\npr{P}\right)+  T_{\textsc{SUM}}\left(n,\npr{P}\right)+1 
    +T_{\textsc{SUM}}\left(n,\frac{\npr{P}}{2}\right)\\
    &<\frac{7n}{\npr{P}}+ 4\log_2 \frac{\npr{P}}{2} +3 T_{\cok{}}\left(\frac{n}{2}, |\mathbf{P}|,M\right)+\frac{6n}{\npr{P}}+ 4\log_2 \frac{\npr{P}}{2}+\frac{7n}{\npr{P}}+ 4\log_2 \npr{P}\\&\quad+\frac{6n}{\npr{P}}+ 4\log_2 \npr{P}+1+\frac{12n}{\npr{P}}+ 4\log_2 \frac{\npr{P}}{2}\\
    &<3 T_{\cok{}}\left(\frac{n}{2}, |\mathbf{P}|,M\right)+38\frac{n}{\npr{P}}+16\log_2 \npr{P}
\end{align*}
Further, in a depth-first recursion level, the \io{} cost of \cok{} (both bandwidth and latency) can be bound by that of the four consecutive invocations of $\cok{}$ used to compute the four subproblems, the cost of redistributing the input (resp., the output) of such subproblems, and the \io{} cost of the three invocations of $\textsc{SUM}$ used to combine the outputs of the three subproblems. We refer the reader to the detailed description of the algorithm. Here we compose the cost of the various operations step-by-step. By Lemma~\ref{lem:costsum} and Lemma~\ref{lem:costdiff}:
\vspace{5mm}
    \begin{align*}
        BW_{\cok{}}\left(n, |\mathbf{P}|, M \right) 
        &\leq  2\times\frac{n}{\npr{P}}+ BW_{\textsc{DIFF}}\left(\frac{n}{2},\frac{\npr{P}}{2} \right)+ 3BW_{\cok{}}\left(\frac{n}{2}, |\mathbf{P}|,M \right)\\&\quad + BW_{\textsc{SUM}}\left(n,\npr{P} \right)+ BW_{\textsc{DIFF}}\left(n,\npr{P}\right) 
        +2\times \frac{n}{\npr{P}}\\&\quad+BW_{\textsc{SUM}}\left(n,\npr{P}\right)+2 +BW_{\textsc{SUM}}\left(n,\npr{P}/2 \right)+3\times \frac{n}{\npr{P}}\nonumber\\
        &\leq 10\frac{n}{\npr{P}} +5\log_2 \frac{\npr{P}}{2} + 3BW_{\cok{}}\left(\frac{n}{2}, |\mathbf{P}|,M \right)\\&\quad+4\log_2 \npr{P}+5\log_2 \npr{P}+4\log_2 \npr{P}+2+4\log_2 \frac{\npr{P}}{2}\nonumber\\
        &\leq 3BW_{\cok{}}\left(\frac{n}{2}, |\mathbf{P}|,M \right) +10\frac{n}{\npr{P}}+22\log_2 \npr{P}\nonumber\\
        \end{align*}
        \begin{align*}
        L_{\cok{}}\left(n, |\mathbf{P}|, M \right) 
        &\leq  2+ L_{\textsc{DIFF}}\left(\frac{n}{2},\frac{\npr{P}}{2}\right)+ 3L_{\cok{}}\left(\frac{n}{2}, |\mathbf{P}|,M \right) + L_{\textsc{SUM}}\left(n,\npr{P} \right)\\&\quad+ L_{\textsc{DIFF}}\left(n,\npr{P}\right)
        +2+L_{\textsc{SUM}}\left(n,\npr{P}\right)+1 +L_{\textsc{SUM}}\left(n,\npr{P}/2\right)+3\nonumber\\
        &\leq 8 +3\log_2 \frac{\npr{P}}{2} + 3L_{\cok{}}\left(\frac{n}{2}, |\mathbf{P}|,M \right)+2\log_2 \npr{P}+3\log_2 \npr{P}\\&\quad+2\log_2 \npr{P}+2\log_2 \frac{\npr{P}}{2}\nonumber\\
        &\leq 3L_{\cok{}}\left(\frac{n}{2}, |\mathbf{P}|,M \right)+15\log_2 \npr{P}\nonumber
    \end{align*}
After $1\le\ell \leq \lceil\log_2 \frac{20n}{M\npr{P}^{\log_3 2}}\rceil$ depth-first steps, the generated sub-problems have input size at most $M\npr{P}^{\log_3 2}/20$, and \cok{} switches to the MI execution mode by invoking $\cok{}_{MI}$. Hence, by Theorem~\ref{thm:costcokUM}, and by the assumptions $n\geq |\mathbf{P}|$, $n\geq M\npr{P}^{\log_32}/10$, and $M\geq \log_2 \npr{P}$  we have:
\begin{align}
    T_{\cok{}}\left(n, |\mathbf{P}|,M\right) 
    &<3 T_{\cok{}}\left(\frac{n}{2}, |\mathbf{P}|,M\right)+38\frac{n}{\npr{P}}+16\log_2 \npr{P}\nonumber\\
    &< 3^{\lceil\log_2 \frac{20n}{M\npr{P}^{\log_3 2}}\rceil}T_{\cok{}}\left(\frac{M\npr{P}^{\log_32}}{20}, |\mathbf{P}|,M\right)+38\frac{n}{\npr{P}}\sum_{i=0}^{\lceil\log_2 \frac{20n}{M\npr{P}^{\log_3 2}}\rceil-1} 2^{-i}\nonumber \\&\quad + 16\log_2 \npr{P}\left(\lceil\log_2 \frac{20n}{M\npr{P}^{\log_3 2}}\rceil-1\right)\nonumber\\
    &< 3\frac{(20n)^{\log_23}}{\npr{P}M^{\log_2 3}}T_{\cok{}_{MI}}\left(\frac{M\npr{P}^{\log_32}}{20}, |\mathbf{P}|\right) + 76\frac{n}{\npr{P}}\nonumber\\&\quad + 16\log_2 \npr{P}\log_2 \frac{20n}{M\npr{P}^{\log_3 2}}\nonumber\\
    &< 3\frac{(20n)^{\log_23}}{\npr{P}M^{\log_2 3}}173\frac{M^{\log_23}\npr{P}}{20^{\log_23}\npr{P}} + 76\frac{n^{\log_23}}{\npr{P}} + 16\log_2 \npr{P}\log_2 \frac{20n}{M\npr{P}^{\log_3 2}}\nonumber\\
    &< 595\frac{n^{\log_23}}{\npr{P}} + 80\frac{n^{\log_23}}{\npr{P}} \label{eq:coktimefina}
\end{align}\vspace{5mm}
\begin{align}
    BW_{\cok{}}\left(n, |\mathbf{P}|, M \right) 
    &\leq 3BW_{\cok{}}\left(\frac{n}{2}, |\mathbf{P}|,M \right) +10\frac{n}{\npr{P}} +22\log_2 \npr{P}\nonumber\\
    &< 3^{\lceil\log_2 \frac{20n}{M\npr{P}^{\log_3 2}}\rceil}BW_{\cok{}}\left(\frac{M\npr{P}^{\log_32}}{20}, |\mathbf{P}|,M\right)+\frac{10n}{\npr{P}}\sum_{i=0}^{\lceil\log_2 \frac{20n}{M\npr{P}^{\log_3 2}}\rceil-1} 2^{-i}\nonumber \\&\quad + 22\log_2 \npr{P}\log_2 \frac{20n}{M\npr{P}^{\log_3 2}}\nonumber\\
    &< 3\frac{(20n)^{\log_23}}{\npr{P}M^{\log_2 3}}BW_{\cok{}_{MI}}\left(\frac{M\npr{P}^{\log_32}}{20}, |\mathbf{P}|\right) + 20\frac{n}{\npr{P}}\nonumber\\&\quad + 22\log_2 \npr{P}\log_2 \frac{20n}{M\npr{P}^{\log_3 2}}\nonumber\\
    &< 3\frac{(20n)^{\log_23}}{\npr{P}M^{\log_2 3}}174\frac{M\npr{P}^{\log_3 2}}{20\npr{P}^{\log_3 2}} + 20\frac{n}{\npr{P}} + 22\log_2 \npr{P}\log_2 \frac{20n}{M\npr{P}^{\log_3 2}}\nonumber\\
    &< (1578+20+110)\left(\frac{n}{M}\right)^{\log_23}\frac{M}{\npr{P}} \label{eq:cokbandfina}
\end{align}
\begin{align}
    L_{\cok{}}\left(n, |\mathbf{P}|,M\right) 
    &\leq 3L_{\cok{}}\left(\frac{n}{2}, |\mathbf{P}|,M \right)+15\log_2 \npr{P}\nonumber\\
    &\leq 3^{\lceil\log_2 \frac{20n}{M\npr{P}^{\log_3 2}}\rceil}L_{\cok{}}\left(\frac{M\npr{P}^{\log_32}}{20}, |\mathbf{P}|,M\right)\nonumber\\&\quad+  15\log_2 |\mathbf{P}|\left(\lceil\log_2 \frac{20n}{M\npr{P}^{\log_3 2}}\rceil-1\right)\nonumber\\
    &< 3\frac{(20n)^{\log_23}}{\npr{P}M^{\log_2 3}}L_{\cok{}_{MI}}\left(\frac{M\npr{P}^{\log_32}}{20}, |\mathbf{P}|\right)+ 15\log_2 |\mathbf{P}|\log_2 \frac{20n}{M\npr{P}^{\log_3 2}}\nonumber\\
    &\leq  3\frac{(20n)^{\log_23}}{\npr{P}M^{\log_2 3}}25\log_2^2|\mathbf{P}| + 15\log_2 |\mathbf{P}|\log_2 \frac{20n}{M\npr{P}^{\log_3 2}}\nonumber\\
    &\leq  8728\frac{n^{\log_23}}{\npr{P}M^{\log_2 3}}\log_2^2|\mathbf{P}|\label{eq:idgaf}
\end{align}
where~\eqref{eq:coktimefina},~\eqref{eq:cokbandfina} and~\eqref{eq:idgaf}  follow as, under the assumptions $n\geq \npr{P}$, $M\geq\log_2 \npr{P}$, and as we are considering the case $n\geq 10M\npr{P}^{\log_32}$,
we have $\frac{n}{\npr{P}}, \log_2 \npr{P}\log_2 \frac{20n}{M\npr{P}^{\log_3 2}}\leq \BO{\left(\frac{n}{M}\right)^{\log_23}\frac{M}{\npr{P}}}$.
\end{proof}

\subsection{Comparison with communication lower bounds}\label{sec:karlwbcomp}
Based on the analysis of \cok{} performance presented in Theorem~\ref{thm:costcokUM} and Theorem~\ref{thm:cok}, we have:
\begin{reptheorem}{thm:informalcok}
\cok{} achieves optimal computation time speedup and optimal bandwidth cost among all parallel Karatsuba-based integer multiplication algorithms. It also minimizes the latency cost up to a $\BO{\log^2 \mathcal{P}}$ multiplicative factor, where $\mathcal{P}$ denotes the number of processors used in the computation.
\end{reptheorem}
\begin{proof}
By Theorem~\ref{thm:costcokUM}, for $M\geq 10n/\mathcal{P}^{\log_3 2}$, the product $C=A\times B$ can be computed using $\cok{}_{MI}$. Under the assumption $n\geq \mathcal{P}$, the bandwidth cost of $\cok{}_{MI}$ asymptotically matches the memory-independent lower bound in Theorem~\ref{thm:karaboundmi}, and its latency latency is within a $\BO{\log^2 \mathcal{P}}$ factor of the corresponding lower bound. Note that the initial distribution of the input values among the processors used in $\cok{}_{MI}$ satisfies the balanced input distribution assumption used to derive  Theorem~\ref{thm:lwbstamemind}.

For $ 10n/\mathcal{P}^{\log_3 2}<M\geq 40n/\mathcal{P}$, by Theorem~\ref{thm:cok}, the product $C=A\times B$ can be computed using $\cok{}$. For $n\geq \mathcal{P}$ and $M\geq \log_2 \mathcal{P}$, the bandwidth cost of $\cok{}_{MI}$ asymptotically matches the memory-dependent lower bound in Theorem~\ref{thm:karabound}, and its latency latency is within a $\BO{\log^2 \mathcal{P}}$ factor of the corresponding lower bound. The total memory space required across the available processors for the execution of \cok{} is $\BO{n}$, that is, within a constant factor of the minimum space required to store the input (and output) values. Finally, in both cases, \cok{} achieves optimal speedup of the computational time $\BO{n^{\log_2 3}/\mathcal{P}}$.
\end{proof}

\section{Conclusion}
We presented parallel algorithms for computing the product of integer numbers in the distributed memory setting. Our algorithm \coim{} is based on the recursive long integer multiplication, while \cok{} is a parallel implementation of Karatsuba's fast multiplication scheme. Under mild conditions on the input size $n$, the number of available processors $P$, and the size of the local cache available to each of them, our algorithms achieve asymptotically optimal computational speedup and bandwidth cost, while their latency cost is within a $\BO{\log^2 P}$ multiplicative factor of the respective theoretical lower bounds. Further, our algorithms require that space available across the processors to be within a multiplicative constant factor of the minimum amount required to store the input values. Due to the common underlying strategy used to obtain both \coim{} and \cok{}, it is possible to combine them seamlessly, thus achieving hybridization of the two algorithmic schemes (as discussed in~\cite{stefani2019io}). Such hybridization is of actual practical interest, as due to the constant factor terms in the complexity characterization of the algorithms (both computational and \io{}) \cok{} allows for overall improved performance over \coim{} for large input size, while when multiplying integers with fewer digits, \coim{} may actually achieve lower execution time.

Due to the large constant factors in the bounds, the presented algorithms are mostly of theoretical interest. While such coefficients can be considerably reduced for particular, and reasonable, values of $n,M$, and $\mathcal{P}$, the pursuit of improved algorithms to be used successfully in practice is an important natural direction for future research. Further, the $\BO{\log^2 P}$ multiplicative factor discrepancy between the latency of our proposed algorithms and the corresponding lower bound leaves open the question on whether it is actually possible to obtain algorithms with lower latency, or if instead, it is possible to obtain \io{} lower bounds which capture such a higher latency requirement. Finally, we believe that the approach discussed in this work could be used to obtain a communication-optimal parallel version of other integer multiplication algorithms, among whom, in particular, the general Toom-Cook-$k$ algorithm.

\bibliographystyle{plain}
\bibliography{bibliography}
\end{document}